\documentclass[11pt]{article}
\usepackage{graphicx}
\usepackage{subcaption}
\usepackage{soul}
\usepackage[colorlinks]{hyperref}
\usepackage[usenames,dvipsnames]{xcolor} 
\usepackage[english]{babel}
\usepackage{url,subcaption,latexsym,amsmath,amssymb,amsfonts,enumitem,epsfig, graphics,times, amsthm, mathtools, bbm}
\usepackage{booktabs}
\usepackage{siunitx}
\usepackage{tikz}
\usepackage{pgfplots}
\usepackage[normalem]{ulem}

\usepackage{algorithm}
\usepackage{algpseudocode}

\usepackage[margin=1in]{geometry}

\newtheorem{theorem}{Theorem}[section]
\newtheorem{corollary}[theorem]{Corollary}

\newtheorem{remark}[theorem]{Remark}

\newtheorem{example}[theorem]{Example}

\def\E{\mathbb{E}}

\title{A general-purpose multi-path stacked coupling for stochastic reaction networks}

\author{
David F. Anderson\thanks{Department of Mathematics, University of
Wisconsin-Madison, USA.  anderson@math.wisc.edu.}
\and 
Jingyi Ma\thanks{Department of Mathematics, University of
Wisconsin-Madison, USA.  jma276@wisc.edu}
}

\begin{document}

\maketitle

\begin{abstract}
Stochastic reaction networks are continuous-time Markov chain models for interacting populations, with applications in biochemistry, epidemiology, ecology, and related areas. Coupling multiple stochastic reaction networks effectively is a fundamental challenge in stochastic simulation. Existing variance-reducing coupling methods are typically designed
for pairs of processes, though specialized multi-path constructions
have been developed for particular applications. To fill this gap, we introduce the multi-path stacked coupling (MSC), a
space-time Poisson construction that jointly simulates a
finite collection of stochastic reaction networks with common
reaction channels and reaction vectors, but possibly different
intensity functions and initial conditions.
Each pairwise marginal of MSC has the same law as the corresponding split-coupled pair, allowing existing variance bounds to transfer directly; in finite-state settings, we also obtain first-order expansions for the mean and second moment of finite-difference numerators. We demonstrate the framework through three sensitivity estimation tasks requiring several nearby parameterized paths: estimating many first derivatives simultaneously, estimating a single first derivative using a wider finite-difference stencil, and estimating higher-order derivatives. Numerical experiments on a processive phosphorylation network demonstrate strong performance in each of the three application areas considered, consistent with the theoretical advantages of MSC: across all three applications, MSC achieves the smallest root mean square error (RMSE) among the methods considered over the tested computational budgets. 
\end{abstract}

     

\section{Introduction}

Stochastic reaction networks, commonly modeled as continuous-time Markov chains on nonnegative integer lattices, provide a natural framework for systems of interacting populations. They are widely used for the quantitative study of intracellular biochemical systems, with applications including gene regulatory networks \cite{arkin1998stochastic,elowitz2002stochastic,mcadams1997stochastic}, viral kinetics \cite{Yin2002}, signal transduction pathways \cite{alon2019introduction,kholodenko2006cell}, and metabolic networks \cite{palsson2006systems}. The same continuous-time Markov jump structure also appears in stochastic population, ecological, and epidemic models \cite{allen2010introduction, andersson2012stochastic}. 

In stochastic reaction network simulation, suitable couplings can keep trajectories close and reduce the variance of Monte Carlo estimators built on their differences. The choice of coupling can therefore have a substantial impact on computational efficiency. Multiple coupling-based strategies have been developed for chemical reaction networks to improve efficiency, including the classical common random number method \cite{law2007simulation}, the common reaction path method \cite{rathinam2010efficient}, split coupling (also called coupled finite differences) \cite{anderson2012efficient,anderson2014asymptotic}, and the stacked coupling method \cite{anderson2019low}. In a different direction, Wang and Rathinam studied the efficiency of Girsanov-transform
and some of these finite-difference sensitivity estimators
\cite{wang2016efficiency}.
 Pathwise-differentiation methods have also been studied. In particular, Sheppard et al.~proposed a regularized pathwise method that is applicable to a subset of stochastic chemical models \cite{sheppard2012pathwise}, whereas Anderson and Wolf proposed a hybrid pathwise differentiation method that is unbiased and applicable to a broad class of stochastic chemical models \cite{wolf2015hybrid}. The hybrid pathwise method is close in spirit to the multi-level Monte Carlo methods developed by Giles \cite{giles2008multilevel} and extended to the stochastic reaction network setting by Anderson and Higham \cite{anderson2012multilevel}. Together, these developments illustrate the importance of coupling as a computational strategy beyond any single estimation task.

Most of these coupling methods are designed for pairwise constructions, though there are important special-purpose extensions to more complicated settings, such as CRP-style shared-path constructions and the double-coupled finite difference method for second derivatives \cite{rathinam2010efficient,wolf2012finite}. However, many important computational tasks require several paths to be simulated  simultaneously. Perhaps the clearest example is the simultaneous estimation of $K>1$ first derivatives, where the most straightforward extension of the existing methods is simply to repeat the basic algorithm $K$ times. Similar issues arise when wider finite-difference stencils are used for the estimation of a single first derivative, or when higher-order derivatives are sought. This leads to the central question of the present paper:

\vspace{.05in}

\noindent \textbf{Motivating question}: Is it possible to construct an efficient joint coupling of multiple stochastic reaction networks sharing the same reaction-channel structure?

\vspace{.05in}

In this paper, we address this question by extending the stacked coupling method of Anderson and Yuan \cite{anderson2019low} to a genuinely multi-path setting. The resulting method, which we term the \emph{multi-path stacked coupling} (MSC), provides a general-purpose framework for coupling a finite collection of stochastic reaction networks sharing a common reaction-channel structure.  Thus, our goal is not merely to handle one additional special case, but rather to provide a clean organizational framework that applies across a range of computational settings.

To keep the discussion concrete, we work throughout with the standard stochastic model for a chemical reaction network: a continuous-time Markov chain $X = \{X(t):t\geq 0\}$ on $\mathbb{Z}_{\geq 0}^d$, with $R$ reaction channels, reaction vectors $\zeta_\ell$, and intensities $\lambda_\ell$, $\ell\in\{1,\dots,R\}$.
This process has infinitesimal generator $\mathcal{L} f(x) = \sum_{\ell=1}^R \lambda_\ell(x)\bigl(f(x+\zeta_\ell)-f(x)\bigr),$
and its forward equation, often termed the chemical master equation, is
\begin{align*}
    \frac{d}{dt}p(t,x)
    =
    \sum_{\ell=1}^R \lambda_\ell(x-\zeta_\ell)p(t,x-\zeta_\ell)
    -
    \sum_{\ell=1}^R \lambda_\ell(x)p(t,x),
\end{align*}
where $p(t,x)=\mathbb{P}(X(t)=x)$; see, for example, \cite{anderson2011continuous, anderson2015stochastic}. 

Under stochastic mass-action kinetics, if reaction $\ell$ consumes
$\nu_{i\ell}$ molecules of species $i$, then its intensity has the form $\lambda_\ell(x) = \kappa_\ell
\prod_{i=1}^d (x_i)_{\nu_{i\ell}},$
where $(x)_m=x(x-1)\dots(x-m+1)$ and $\kappa_\ell$ is the corresponding rate constant.

For applications of MSC, parametric sensitivity analysis provides a particularly clear setting for demonstrating our coupling framework. Stochastic reaction network models depend on a number of parameters, and it is therefore natural to ask how outputs of interest respond to perturbations of those parameters. Parametric sensitivity analysis quantifies this dependence through derivatives of expected model outputs with respect to model parameters, such as the expected abundance of a reactant at a fixed time or the probability of a specified model behavior \cite{anderson2012efficient,gupta2013unbiased}. 
For a parameterized family $X^\theta$, with
$\theta\in\mathbb{R}^K$, define
$g(\theta)=\mathbb{E}[f(X^\theta(T))]$, where
$f:\mathbb{Z}_{\geq0}^d\to\mathbb{R}$ is an observable
and $T>0$ is fixed.
We are interested in computational settings where finite-difference approximations require the simultaneous generation of several nearby parameterized processes, such as 
\begin{align*}
X^\theta,\quad X^{\theta+\varepsilon_1 e_1},\quad \dots\quad ,\quad X^{\theta+\varepsilon_K e_K},
\end{align*}
where $e_i$ is the $i$th canonical basis vector in $\mathbb R^K$ and each $\varepsilon_i$ is a small perturbation size.

The key point is that existing coupling methods involve different tradeoffs in such settings. Repeated application of the split coupling yields favorable pairwise variance behavior, but it requires separate coupled pairs to be generated for each perturbation direction. In contrast, the common reaction path constructions allow randomness to be shared across many perturbed systems simultaneously, thereby reducing the number of simulated trajectories required per Monte Carlo replication, though typically at the cost of weaker variance control. The MSC framework is designed to combine the main advantages of these two viewpoints: it generates all required nearby paths simultaneously, with the same favorable path-count structure as the common reaction path method, while retaining the strong pairwise coupling behavior associated with the split coupling.

We highlight three main application areas, which will be  developed further in Section \ref{sec:FD_settings} before being numerically studied in Sections \ref{Motivation 1}, \ref{Motivation 2}, and \ref{Motivation 3}.

\begin{enumerate}
    \item \textbf{Simultaneous computation of many first derivatives.}
    
    This is the clearest and most important application. Suppose one wishes to estimate all $K$ first-order sensitivities simultaneously.
    In a single Monte Carlo replication, repeated use of split coupling requires the generation of $2K$ sample paths, since one coupled nominal/perturbed pair must be generated for each derivative direction. In contrast, the common reaction path method requires only $K+1$ sample paths per replication: one nominal path together with the $K$ perturbed paths. The MSC framework has this same $K+1$ scaling per replication, while retaining the favorable pairwise variance properties associated with the split coupling. Thus, when many derivatives are required simultaneously, MSC combines the main organizational advantage of CRP with the variance reduction of split coupling.

    \item \textbf{Higher-order finite-difference estimation of a single first derivative.}
    
    Even when estimating a single first derivative, improved finite-difference formulas often require wider stencils and hence multiple nearby parameter values. MSC provides a systematic way to couple the resulting collection of paths within each Monte Carlo replication. This makes it possible to combine the reduced bias of higher-order finite-difference formulas with favorable variance behavior, and hence to obtain improved overall performance.

    \item \textbf{Higher-order derivatives.}
    
    For second-, third-, and higher-order derivatives, as well as mixed partial derivatives, the relevant finite-difference formulas can involve several nearby parameter values, and constructing effective couplings by hand quickly becomes cumbersome. MSC provides a clean framework for such settings. While specially tailored methods may perform better in particular problems, MSC offers a broadly applicable approach that avoids the need for problem-specific coupling constructions.
\end{enumerate}

The remainder of the paper is organized as follows. In Section~\ref{sec:MSC_framework}, we introduce the multi-path stacked coupling, present the corresponding simulation algorithm, and establish its main theoretical properties. In Section~\ref{sec:FD_settings}, we make precise the three representative finite-difference settings described above, each of which naturally involves multiple nearby parameterized sample paths. In Section~\ref{sec:Variance Bounds and Finite-State Asymptotic Expansions}, we present the variance bounds and refined finite-state asymptotic expansions used in the subsequent analysis. In Section~\ref{General optimal RMSE scaling for finite-difference estimators}, we give a unified bias-variance analysis and derive the optimal RMSE scaling rates for the different path-generation methods. Finally, in Section~\ref{sec:numerics}, we introduce the processive phosphorylation network used throughout the numerical study and apply the different methods to the three finite-difference settings considered in Section~\ref{sec:FD_settings}.

\section{The multi-path stacked coupling framework}
\label{sec:MSC_framework}

We now introduce the multi-path stacked coupling framework. First, in Section~\ref{sec:msc_subsection}, we define the MSC construction, give the corresponding simulation algorithm, and provide a simple illustration of the stacked Poisson-space representation. Second, in Section~\ref{sec:MSC_theory}, we establish the theoretical properties of MSC that will be used in the applications below. In particular, we show that the relevant pairwise marginals under MSC have the same law as the corresponding split-coupled pairs.

\subsection{Multi-path stacked coupling}
\label{sec:msc_subsection}

To define the multi-path stacked coupling method, we first recall a space-time Poisson representation of a stochastic reaction network, which is an alternative to Kurtz's random-time-change representation \cite{kurtz1980representations}. Let $\mathcal{N}$ denote a unit-rate Poisson point process on $[0,\infty)\times [0,\infty)$. First, we consider a finite collection of stochastic reaction networks $X^{(j)}$, $j \in \{0,\cdots,M\}$, sharing a common set of reaction channels indexed by $\ell\in\{1,\dots,R\}$ and common reaction vectors $\zeta_\ell$. Each process has its own intensity functions $\lambda_\ell^{(j)}(x)$ and the processes may also differ in their initial conditions. The intensity functions are time-independent, and we assume the coordinate processes are nonexplosive. If the initial states are random, their joint distribution is specified independently of the driving Poisson point process.

Define
\begin{equation*}
q_\ell^{(j)}(t)
=
\sum_{r=1}^\ell \lambda_r^{(j)}(X^{(j)}(t)),
\qquad
q_0^{(j)}(t)=0.
\end{equation*}
Then $X^{(j)}$ may be represented as the solution to
\begin{align}
\label{cdjskbfdnvsdc}
X^{(j)}(t)
&=
X^{(j)}(0)
+
\sum_{\ell=1}^R \zeta_\ell
\int_{[0,t]\times[0,\infty)}
1_{\left[
q^{(j)}_{\ell-1}(s-),
q^{(j)}_{\ell-1}(s-)
+
\lambda_\ell^{(j)}\left(X^{(j)}(s-)\right)
\right)}(x)
\mathcal{N}(ds\times dx).
\end{align}
This representation gives a continuous-time Markov chain with generator
\begin{equation*}
\mathcal L^{(j)} f(x)
=
\sum_{\ell=1}^R
\lambda_\ell^{(j)}(x)
\left(f(x+\zeta_\ell)-f(x)\right),
\end{equation*}
and hence has the same law as the standard stochastic reaction network with intensity functions $\lambda_\ell^{(j)}$, reaction vectors $\zeta_\ell$, and the same initial distribution.

Equation \eqref{cdjskbfdnvsdc} immediately suggests a coupling of the processes $X^{(0)}$, $X^{(1)}$, $\cdots$, $X^{(M)}$: use the same realization of $\mathcal{N}$  for every $j$.   However, this coupling has an important limitation: the same Poisson point may fall in different reaction-channel intervals for different processes. Thus, simultaneous jumps need not occur through the same reaction channel \cite{anderson2019low}.

We now develop the multi-path stacked coupling, which fixes this issue by assigning a common region of Poisson space to each reaction channel.
For the finite collection of $X^{(j)}$, $j \in \{0,\cdots,M\}$, we define the channel-wise envelope rates 
\[
\bar\lambda_\ell(s)
=
\sup_{0\leq j\leq M}
\lambda_\ell^{(j)}\left(X^{(j)}(s)\right),
\qquad
\bar q_\ell(s)=\sum_{r=1}^\ell \bar\lambda_r(s),
\qquad
\bar q_0(s)=0.
\]
In the stacked construction, reaction channel $\ell$ is assigned the vertical strip
\[
\left[\bar q_{\ell-1}(s),\bar q_{\ell-1}(s)+\bar\lambda_\ell(s)\right),
\]
whose height is $\bar\lambda_\ell(s)$. Within this strip, the process $X^{(j)}$ accepts a Poisson point only if that point falls in the subinterval of height $\lambda_\ell^{(j)}(X^{(j)}(s))$. Consequently, if several processes jump at the same Poisson point, then they necessarily jump through the same reaction channel $\ell$.

The multi-path stacked coupling is therefore defined, for $j\in\{0,\dots,M\}$, by
\begin{align}
\label{representation 123123}
\begin{split}
X^{(j)}(t)
&=
X^{(j)}(0)
+
\sum_{\ell=1}^R \zeta_\ell
\int_{[0,t]\times[0,\infty)}
1_{\left[\bar q_{\ell-1}(s-),\bar q_{\ell-1}(s-)
+\lambda_\ell^{(j)}\left(X^{(j)}(s-)\right)\right)}(x)
\mathcal{N}(ds\times dx),
\end{split}
\end{align}
where $\bar q_\ell$ and $\bar\lambda_\ell$ are as defined above.

We next turn from the stochastic representation \eqref{representation 123123} to the precise numerical algorithm used to simulate it. At each step, the channel-wise envelope rates $\bar\lambda_\ell$ are computed from the current states of all paths, a candidate Poisson point is generated under the total envelope rate, and each path accepts the point according to its own intensity within the selected reaction channel.

\begin{algorithm}[H]
\caption{Multi-path stacked coupling algorithm}
\label{alg:multi-path-stacked-coupling}
\begin{algorithmic}[1]

\Statex Fix a terminal time $T>0$ and initial states $x^{(j)}$
for $j\in\{0,\dots,M\}$. Initialize $t=0$ and
$X^{(j)}=x^{(j)}$ for all $j$. Repeat the following steps
while $t<T$. All random variables generated below are
independent of those generated previously.

   \State Compute all intensities $\lambda_\ell^{(j)}\left(X^{(j)}\right)$ for all $\ell \in\{1, \dots, R\}$ and $j \in \{0,\dots,M\}$. Then compute all $\bar{q}_\ell=\sum_{r=1}^\ell \bar{\lambda}_r$ with $\bar{\lambda}_r=\sup _j\left\{\lambda_r^{(j)}\left(X^{(j)}\right)\right\}$.
    
    \State Generate a unit exponential random variable $\tau$  and set $\Delta=\frac{\tau}{\bar{q}_R}$. (If $\bar{q}_R = 0$, set $\Delta = \infty$.)
    
    \If {$t+\Delta>T$}
        \State Set $X^{(j)}(u)=X^{(j)}$ for all $j \in \{0,\dots,M\}$ and $u \in[t,T]$.
        \State Set $t=T$ and end the loop.
    \Else
        \State Set $X^{(j)}(u)=X^{(j)}$  for all $j \in \{0,\dots,M\}$ and $u \in[t,t+\Delta)$.
        \State Generate a uniform random variable $\xi$ on $[0,1]$.
            \State Find $\mu \in\{1, \dots, R\}$ for which
            $\frac{\bar{q}_{\mu-1}}{\bar{q}_R} \leq \xi < \frac{\bar{q}_\mu}{\bar{q}_R}$.
            \ForAll{$j \in \{0,\dots,M\}$}
            \If{
            $\frac{\bar{q}_{\mu-1}}{\bar{q}_R} \leq \xi < \frac{\bar{q}_{\mu-1}+\lambda_\mu^{(j)}\left(X^{(j)}\right)}{\bar{q}_R}$
            }
            \State $X^{(j)} \leftarrow X^{(j)}+\zeta_\mu$.
            \EndIf 
            \EndFor 
        \State Set $t \leftarrow t+\Delta$ and return to Step 1.
    \EndIf 
\end{algorithmic}
\end{algorithm}

We close with a simple example illustrating how the stacked Poisson-space construction determines which processes accept a given Poisson point.

\begin{example}
We consider a reaction network with two species and two reactions,
\begin{align*}
\text{Reaction 1:} & \quad 2A \xrightarrow{k_1} B,
\qquad \zeta_1=(-2,1),\\
\text{Reaction 2:} & \quad B \xrightarrow{k_2} 2A,
\qquad \zeta_2=(2,-1).
\end{align*}
For a parameter vector $\theta=(k_1,k_2)$, the intensity functions are
\begin{align*}
\lambda_1^\theta(x_A,x_B) 
= k_1 x_A(x_A-1), \quad 
\lambda_2^\theta(x_A,x_B)
= k_2 x_B.
\end{align*}
We compare three parameter vectors,
\[
    \theta_0=(20,20),\qquad
    \theta_1=(25,25),\qquad
    \theta_2=(30,30).
\]

Set $X^{(j)}=X^{\theta_j}$ and
$\lambda_\ell^{(j)}=\lambda_\ell^{\theta_j}$ for
$j\in\{0,1,2\}$. Thus this example is an MSC construction
with $M=2$.

We use the common initial condition $X^{(0)}(0)=X^{(1)}(0)=X^{(2)}(0)=(20,200).$
To make the construction explicit, suppose that the underlying space-time Poisson point process has two relevant points, one at time $t=1$ with vertical coordinate $3000$, and one at time $t=2$ with vertical coordinate $14800$. We now determine which processes accept these points under the MSC construction.

At the initial state $(20,200)$, the reaction-1 intensities are
\[
\lambda_1^{(0)}=20\cdot20\cdot19=7600,\quad
\lambda_1^{(1)}=25\cdot20\cdot19=9500,\quad
\lambda_1^{(2)}=30\cdot20\cdot19=11400,
\]
and the reaction-2 intensities are
\[
\lambda_2^{(0)}=20\cdot200=4000,\qquad
\lambda_2^{(1)}=25\cdot200=5000,\qquad
\lambda_2^{(2)}=30\cdot200=6000.
\]
Taking the maximum over the three processes in each reaction channel yields
\[
\bar\lambda_1=11400,\qquad \bar\lambda_2=6000.
\]
Thus, before the first jump, reaction channel $1$ occupies the vertical interval $[0,11400)$ and reaction channel $2$ occupies the vertical interval $[11400,17400)$.

The first Poisson point, which we recall is assumed to be at $(1,3000)$, has vertical coordinate $3000$. Since $3000<11400$, the point lies in the reaction-1 region. Moreover,
\[
3000<7600,\qquad 3000<9500,\qquad 3000<11400,
\]
so all three processes accept this point as a reaction-1 event. Therefore all three processes are updated by $\zeta_1=(-2,1)$, giving
\[
X^{(0)}(1)=X^{(1)}(1)=X^{(2)}(1)=(18,201).
\]

Before the second Poisson point, all three processes are therefore at state $(18,201)$. The reaction-1 intensities are
\[
\lambda_1^{(0)}=20\cdot18\cdot17=6120,\qquad
\lambda_1^{(1)}=25\cdot18\cdot17=7650,\qquad
\lambda_1^{(2)}=30\cdot18\cdot17=9180,
\]
and the reaction-2 intensities are
\[
\lambda_2^{(0)}=20\cdot201=4020,\qquad
\lambda_2^{(1)}=25\cdot201=5025,\qquad
\lambda_2^{(2)}=30\cdot201=6030.
\]
Hence
\[
\bar\lambda_1=9180,\qquad \bar\lambda_2=6030.
\]
Thus reaction channel $1$ occupies the vertical interval $[0,9180)$, while reaction channel $2$ occupies the vertical interval $[9180,15210)$.

The second Poisson point, which we recall is assumed to be at $(2,14800)$, has vertical coordinate $14800$. Since
\[
9180<14800<15210,
\]
this point lies in the reaction-2 region. Within that region, the acceptance intervals for the three processes are
\[
[9180,13200) \quad \text{for } X^{(0)},\quad
[9180,14205) \quad \text{for } X^{(1)},\quad
[9180,15210) \quad \text{for } X^{(2)}.
\]
Since $14205<14800<15210$, only $X^{(2)}$ accepts the point. Consequently,
\[
X^{(2)}(2)=(18,201)+(2,-1)=(20,200),
\]
while $X^{(0)}$ and $X^{(1)}$ remain unchanged at $(18,201)$.

Figure~\ref{fig1} gives a visual representation of this construction for the two Poisson points considered above. \hfill $\triangle$
\end{example}

\begin{figure}
    \centering

\begin{tikzpicture}
\begin{axis}[
width = 10cm,
height = 6cm,
title = {Multi-path Stacked Coupling Illustration},
xmin = 0, 
xmax= 3,
ymin = 0,
xlabel = {Time},
ylabel = {Stacked intensities},
axis x line= bottom,
axis y line= left,
ticklabel style={/tikz/font=\small},
legend style ={
at = {(1.02,1)},
anchor = north west,
font = \small
}
]

\addplot[const plot, dashed, thick, green!60!black]
    coordinates {
        (0,7600) (1,7600)
        (1,6120) (2,6120)
        (2,6120) (3,6120)
    };

\addplot[const plot, dashed, thick, red]
    coordinates {
        (0,9500) (1,9500)
        (1,7650) (2,7650)
        (2,7650) (3,7650)
    };

\addplot[const plot, dashed, thick, color=purple]
    coordinates {
        (0,11400) (1,11400)
        (1,9180) (2,9180)
        (2,11400) (3,11400)
    };

\addplot[const plot, thick, color=red!40]
    coordinates {
        (0,11500) (1,11500)
        (1,9280) (2,9280)
        (2,11500) (3,11500)
    };

\addplot[const plot, dashed, thick, color=cyan]
    coordinates {
        (0,17400) (1,17400)
        (1,15210) (2,15210)
        (2,17400) (3,17400)
    };

\addplot[const plot, dashed, thick, color=brown]
    coordinates {
        (0,15400) (1,15400)
        (1,13200) (2,13200)
        (2,15420) (3,15420)
    };

\addplot[const plot, dashed, thick, red!20!black]
    coordinates {
        (0,16400) (1,16400)
        (1,14205) (2,14205)
        (2,16425) (3,16425)
    };

\addplot[const plot, thick, color=cyan!40]
    coordinates {
        (0,17500) (1,17500)
        (1,15310) (2,15310)
        (2,17500) (3,17500)
    };

\addplot[only marks, mark=x, mark size=3pt, thick]
    coordinates {(1,3000) (2,14800)};

\legend{
    $\lambda_1^{(0)}(t)$,
    $\lambda_1^{(1)}(t)$,
    $\lambda_1^{(2)}(t)$,
    $\bar\lambda_1(t)$,
    $\bar\lambda_1(t)+\lambda_2^{(2)}(t)$,
    $\bar\lambda_1(t)+\lambda_2^{(0)}(t)$,
    $\bar\lambda_1(t)+\lambda_2^{(1)}(t)$,
    $\bar\lambda_1(t)+\bar\lambda_2(t)$
}
\end{axis}
\end{tikzpicture}
    \caption{Illustration of the multi-path stacked coupling for the three processes $X^{(0)}$, $X^{(1)}$, and $X^{(2)}$. The lower region corresponds to reaction channel $1$, with upper boundary $\bar\lambda_1(t)$. The upper region corresponds to reaction channel $2$, with upper boundary $\bar\lambda_1(t)+\bar\lambda_2(t)$. Dashed lines within each region show the acceptance boundaries for the individual processes. The Poisson point at $(1,3000)$ is accepted by all three processes as a reaction-1 event, while the Poisson point at $(2,14800)$ is accepted only by $X^{(2)}$ as a reaction-2 event. The two solid envelope boundaries are shifted upward by $100$ units for visibility.}
    \label{fig1}
\end{figure}
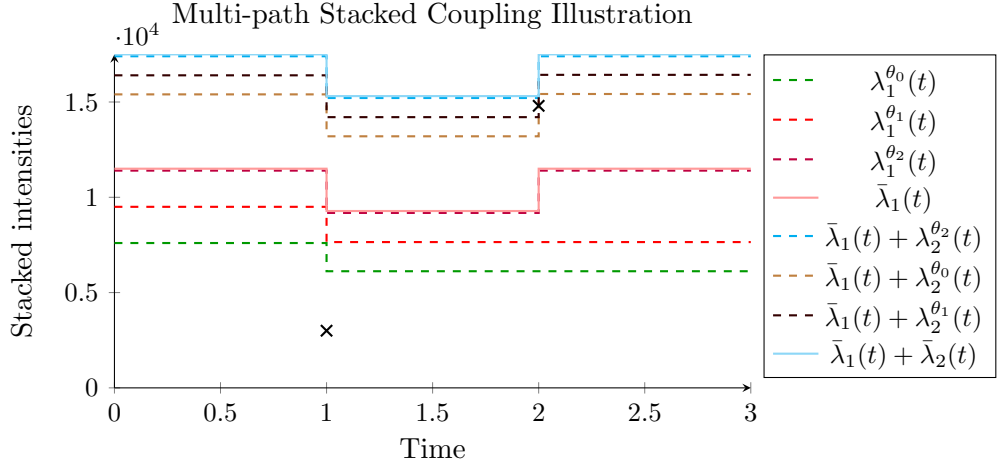

\subsection{Theoretical properties of MSC}
  \label{sec:MSC_theory}

  In this subsection, we establish the main theoretical properties of MSC.
  Each relevant pairwise marginal of the MSC construction has the same law as the corresponding split-coupled pair.  This allows us to transfer existing variance bounds for split coupling directly to MSC.

Throughout, the subscript $\mathrm{MSC}$ denotes processes generated by the multi-path stacked coupling, while the subscript $\mathrm{SC}$ denotes processes generated by the split coupling. We also write $a\wedge b=\min\{a,b\}$.

The split coupling, also known as the coupled finite-difference method, is a standard coupling for two stochastic reaction networks sharing the same reaction channels and reaction vectors \cite{anderson2012efficient}. The resulting pair is itself a continuous-time Markov chain on the product state space: for each reaction channel, the two coordinates share jumps at the minimum of their two intensities, while the remaining excess intensity drives jumps only in the coordinate with the larger intensity. 

For fixed distinct $i,j\in \{0,\dots,M\}$ and each $\ell \in \{1, \dots, R\}$, define 
\[
m_\ell(s)=
\lambda_\ell^{(i)}\left(X_{\mathrm{SC}}^{(i)}(s)\right)
\wedge
\lambda_\ell^{(j)}\left(X_{\mathrm{SC}}^{(j)}(s)\right).
\]
Then the split-coupled pair can be represented as
\begin{align}
\label{split coupling}
\begin{split}
X_{\mathrm{SC}}^{(i)}(t)
&=
X_{\mathrm{SC}}^{(i)}(0)
+
\sum_{\ell=1}^R
Y_{\ell,1}\left(\int_0^t m_\ell(s)\,ds\right)\zeta_\ell +
\sum_{\ell=1}^R
Y_{\ell,2}\left(
\int_0^t
\bigl(
\lambda_\ell^{(i)}\left(X_{\mathrm{SC}}^{(i)}(s)\right)-m_\ell(s)
\bigr)\,ds
\right)\zeta_\ell,\\
X_{\mathrm{SC}}^{(j)}(t)
&=
X_{\mathrm{SC}}^{(j)}(0)
+
\sum_{\ell=1}^R
Y_{\ell,1}\left(\int_0^t m_\ell(s)\,ds\right)\zeta_\ell +
\sum_{\ell=1}^R
Y_{\ell,3}\left(
\int_0^t
\bigl(
\lambda_\ell^{(j)}\left(X_{\mathrm{SC}}^{(j)}(s)\right)-m_\ell(s)
\bigr)\,ds
\right)\zeta_\ell,
\end{split}
\end{align}
where the $Y_{\ell,r}$, for $\ell\in\{1,\dots,R\}$ and
$r\in\{1,2,3\}$, are independent unit-rate Poisson processes;
see \cite{anderson2012efficient}. The initial states $X_{\mathrm{SC}}^{(i)}(0)$ and
$X_{\mathrm{SC}}^{(j)}(0)$ need not coincide, and the initial
pair is independent of the driving Poisson processes.

The following theorem shows that, for any two distinct coordinates $i,j\in \{0,\dots,M\}$, the MSC marginal pair $\left(X^{(i)},X^{(j)}\right)$ has exactly the same distribution as the corresponding split-coupled pair with the same initial joint distribution, provided the split-coupled pair and MSC marginal pair are nonexplosive.

\begin{theorem}
\label{equal theorem111}
Fix distinct $i,j\in\{0,\dots,M\}$. Let
$\left(X_{\mathrm{SC}}^{(i)},X_{\mathrm{SC}}^{(j)}\right)$
be the split-coupled pair generated by \eqref{split coupling}, and let
$\left(X_{\mathrm{MSC}}^{(0)},X_{\mathrm{MSC}}^{(1)},\dots,X_{\mathrm{MSC}}^{(M)}\right)$
be the multi-path process generated by the MSC representation \eqref{representation 123123}. Suppose that the split-coupled pair and the MSC marginal pair
$\left(X_{\mathrm{MSC}}^{(i)},X_{\mathrm{MSC}}^{(j)}\right)$
are nonexplosive and have the same initial joint distribution. Then
\begin{equation*}
\left(X_{\mathrm{SC}}^{(i)},X_{\mathrm{SC}}^{(j)}\right)
\overset{d}{=}
\left(X_{\mathrm{MSC}}^{(i)},X_{\mathrm{MSC}}^{(j)}\right).
\end{equation*}
\end{theorem}

\begin{proof}
It suffices to show that the MSC marginal pair is a continuous-time Markov chain with the same generator as the split-coupled pair. Indeed, since the two pair processes have the same initial distribution and are nonexplosive, equality of their generators implies equality in law by the standard uniqueness theorem for nonexplosive continuous-time Markov chains on countable state spaces; see, for example, Theorem 2.8.4 of \cite{norris1998markov}.

Let $z=(x_i,x_j)\in \mathbb Z_{\geq 0}^d\times \mathbb Z_{\geq 0}^d$,
and define $\lambda_{\ell,i,j}^{\min}(z)=
\lambda_\ell^{(i)}(x_i)
\wedge
\lambda_\ell^{(j)}(x_j).$
For any bounded function $h:\mathbb Z_{\geq 0}^d\times\mathbb Z_{\geq 0}^d\to\mathbb R$, the split-coupled pair has generator
\begin{align*}
\mathcal L_{i,j} h(z)
=
\sum_{\ell=1}^R
\Bigl\{
&
\lambda_{\ell,i,j}^{\min}(z)
\left(
h\left(z+(\zeta_\ell,\zeta_\ell)\right)-h(z)
\right)
\\
&+
\left(
\lambda_\ell^{(i)}(x_i)-\lambda_{\ell,i,j}^{\min}(z)
\right)
\left(
h\left(z+(\zeta_\ell,0)\right)-h(z)
\right)
\\
&+
\left(
\lambda_\ell^{(j)}(x_j)-\lambda_{\ell,i,j}^{\min}(z)
\right)
\left(
h\left(z+(0,\zeta_\ell)\right)-h(z)
\right)
\Bigr\}.
\end{align*}

We now consider the MSC process and project onto the coordinates $(i,j)$. Fix a full MSC state $z^*=(x_0,x_1,\dots,x_M),$
and write $z=(x_i,x_j)$ for its projection onto coordinates $i$ and $j$. In the reaction-$\ell$ block of the stacked Poisson representation, the two coordinates $X_{\mathrm{MSC}}^{(i)}$ and $X_{\mathrm{MSC}}^{(j)}$ both jump through reaction $\ell$ on a subinterval of height
\begin{equation*}
\lambda_\ell^{(i)}(x_i)\wedge\lambda_\ell^{(j)}(x_j)
=
\lambda_{\ell,i,j}^{\min}(z).
\end{equation*}
Only $X_{\mathrm{MSC}}^{(i)}$ jumps on a subinterval of height
\begin{equation*}
\lambda_\ell^{(i)}(x_i)-\lambda_{\ell,i,j}^{\min}(z),
\end{equation*}
and only $X_{\mathrm{MSC}}^{(j)}$ jumps on a subinterval of height
\begin{equation*}
\lambda_\ell^{(j)}(x_j)-\lambda_{\ell,i,j}^{\min}(z).
\end{equation*}
All remaining points in the reaction-$\ell$ block affect neither coordinate of the projected pair and therefore contribute no transition to the projected process.

Thus, after projection onto the coordinates $(i,j)$, every nontrivial transition rate depends only on the projected state $z=(x_i,x_j)$, not on the remaining coordinates of $z^*$. Events that change only the other coordinates are invisible to the projection. Hence the projected pair
$\left(X_{\mathrm{MSC}}^{(i)},X_{\mathrm{MSC}}^{(j)}\right)$
is itself a continuous-time Markov chain, with infinitesimal generator $\mathcal L_{i,j}$, the same generator as that of the split-coupled pair.

Therefore the MSC marginal pair and the split-coupled pair have the same generator, and the equality in law follows from the uniqueness result cited at the start of the proof.
\end{proof}

\section{Finite-difference settings requiring multiple paths}
\label{sec:FD_settings}

We now specialize the general MSC construction of
Section~\ref{sec:MSC_framework} to finite-difference sensitivity
estimation. For parameter vectors $\theta_0,\dots,\theta_M$,
we take $X^{(j)}=X^{\theta_j}$ and
$\lambda_\ell^{(j)}=\lambda_\ell^{\theta_j}$, with all processes
starting from the same fixed state. With
$g(\theta)=\mathbb{E}[f(X^\theta(T))]$ as in the introduction,
we consider three representative settings and the corresponding
collections of paths required.

We begin with the standard forward finite-difference approximation for a first-order sensitivity. Under the usual smoothness assumptions on $g$, for $j\in\{1,\dots,K\}$,
\begin{align}
\label{eq:forward-fd}
\frac{\partial}{\partial \theta_j}g(\theta)
=
\frac{\mathbb{E}\big[f(X^{\theta+\varepsilon_j e_j}(T))\big]-\mathbb{E}\big[f(X^\theta(T))\big]}{\varepsilon_j}
+ O(\varepsilon_j),
\end{align}
where $e_j$ is the $j$th unit vector in $\mathbb{R}^K$ and $\varepsilon_j$ is the perturbation size. The corresponding Monte Carlo estimator is
\begin{align}
\label{eq:mc-forward-fd}
\widehat{D}_j(\theta)
=
\frac{1}{N}\sum_{i=1}^N d_{[i]}(\varepsilon_j),
\qquad
d_{[i]}(\varepsilon_j)
=
\frac{f(X^{\theta+\varepsilon_j e_j}_{[i]}(T))-f(X^\theta_{[i]}(T))}{\varepsilon_j},
\end{align}
where $\big(X^\theta_{[i]},X^{\theta+\varepsilon_j e_j}_{[i]}\big)$ is the $i$th coupled pair of sample paths.

We now record three representative settings in which finite-difference estimators can require more than a single coupled pair.

\begin{enumerate}
    \item \textbf{Simultaneous computation of many first derivatives.}

    If $\theta$ is multi-dimensional and one wishes to estimate all first-order sensitivities, then for each Monte Carlo replication one is naturally led to the collection of sample paths
    \begin{align}
    \label{eq:many-first-derivatives}
    \left(
    X^\theta_{[i]},
    X^{\theta+\varepsilon_1 e_1}_{[i]},
    \dots,
    X^{\theta+\varepsilon_K e_K}_{[i]}
    \right).
    \end{align}
    Thus, although each individual derivative is approximated by a two-point finite difference, the overall computation requires a multi-path construction.
\item \textbf{Higher-order finite-difference estimation of a single first derivative.}

Even when only a single first derivative is sought, higher-order finite-difference formulas may require wider stencils.

The standard centered approximation
\begin{align}
\label{eq:second-order-first-derivative}
\frac{\partial}{\partial \theta_j}g(\theta)
&=
\frac{
\mathbb{E}\big[f(X^{\theta+\varepsilon_j e_j}(T))\big]
-
\mathbb{E}\big[f(X^{\theta-\varepsilon_j e_j}(T))\big]
}
{2\varepsilon_j}
+O(\varepsilon_j^2)
\end{align}
requires, for each Monte Carlo replication, the two-path collection
\begin{align}
\label{eq:second-order-first-derivative-paths}
\left(
X^{\theta-\varepsilon_j e_j}_{[i]},
X^{\theta+\varepsilon_j e_j}_{[i]}
\right).
\end{align}
A higher-order alternative is the fourth-order centered approximation
\begin{align}
\label{eq:fourth-order-first-derivative}
\frac{\partial}{\partial \theta_j}g(\theta)
&=
\frac{1}{12\varepsilon_j}
\Big[
-\mathbb{E}\big[f(X^{\theta+2\varepsilon_j e_j}(T))\big]
+8\,\mathbb{E}\big[f(X^{\theta+\varepsilon_j e_j}(T))\big] \nonumber\\
&\qquad
-8\,\mathbb{E}\big[f(X^{\theta-\varepsilon_j e_j}(T))\big]
+\mathbb{E}\big[f(X^{\theta-2\varepsilon_j e_j}(T))\big]
\Big]
+ O(\varepsilon_j^4)
\end{align}
which requires, for each Monte Carlo replication, the four-path collection
\begin{align}
\label{eq:wider-stencil-paths}
\left(
X^{\theta+2\varepsilon_j e_j}_{[i]},
X^{\theta+\varepsilon_j e_j}_{[i]},
X^{\theta-\varepsilon_j e_j}_{[i]},
X^{\theta-2\varepsilon_j e_j}_{[i]}
\right).
\end{align}
Thus, wider stencils for a single first derivative also naturally lead to a multi-path coupling problem.

    \item \textbf{Higher-order derivatives.}

    A similar issue arises for higher-order sensitivities. For example, a centered finite-difference approximation for the third derivative is
    \begin{align}
    \label{eq:third-derivative-fd}
    \frac{\partial^3}{\partial \theta_j^3}g(\theta)
    &=
    \frac{1}{2\varepsilon_j^3}
    \Big[
    \mathbb{E}\big[f(X^{\theta+2\varepsilon_j e_j}(T))\big]
    -2\,\mathbb{E}\big[f(X^{\theta+\varepsilon_j e_j}(T))\big] \nonumber\\
    &\qquad
    +2\,\mathbb{E}\big[f(X^{\theta-\varepsilon_j e_j}(T))\big]
    -\mathbb{E}\big[f(X^{\theta-2\varepsilon_j e_j}(T))\big]
    \Big]
    + O(\varepsilon_j^2).
    \end{align}
    This approximation requires the same four-path collection as
\eqref{eq:wider-stencil-paths}. 
More generally, higher-order derivatives
and mixed partial derivatives can require increasingly complicated
collections of nearby parameterized sample paths.
 \end{enumerate}

These settings provide important applications of the general
MSC construction introduced in Section~\ref{sec:MSC_framework}.

\section{Variance bounds and finite-state asymptotic expansions}
\label{sec:Variance Bounds and Finite-State Asymptotic Expansions}

In this section, we present some variance bounds and estimates that will be used repeatedly in the subsequent sections. Notably, we give a finite-state refinement of this variance theory: using a perturbation expansion for the $Q$-matrix of a coupled process, we obtain leading-order expansions for the mean and second moment of finite-difference numerators.  This refinement  is new and provides sharper asymptotic information than the order bounds alone.  Such information can be useful for understanding leading variance constants, choosing perturbation sizes, comparing coupling methods beyond order estimates, and guiding more refined Monte Carlo implementations.

First, we record the variance estimate that will be used repeatedly below.  This is essentially Theorem 3.1 of \cite{anderson2012efficient}, applied to the pairwise marginals of MSC.  In particular, we assume that the hypotheses of that theorem hold for the parameter perturbations considered here: the intensity functions and drift are uniformly Lipschitz on the relevant state space, the change in the intensity functions and drift is $O(\varepsilon)$ under an $O(\varepsilon)$ parameter perturbation, and $f$ is Lipschitz on the relevant state space.

\begin{corollary}
\label{cor:msc-pairwise-variance}
Consider a parameterized family of stochastic reaction networks
$\{X^\theta:\theta \in \mathbb R^K\}$,
sharing common reaction channels and reaction vectors, and
starting from the same fixed initial state.

For each sufficiently small $\varepsilon>0$, choose two
parameter vectors $\alpha_\varepsilon,\beta_\varepsilon$
such that
\begin{equation*}
\|\alpha_\varepsilon-\beta_\varepsilon\|
\leq C_0\varepsilon
\end{equation*}
for some constant $C_0>0$ independent of $\varepsilon$.

Let
$\left(X_{\mathrm{MSC}}^{\alpha_\varepsilon},
X_{\mathrm{MSC}}^{\beta_\varepsilon}\right)$
be two coordinates of an MSC construction, and let
$\left(X_{\mathrm{SC}}^{\alpha_\varepsilon},
X_{\mathrm{SC}}^{\beta_\varepsilon}\right)$
be the corresponding split-coupled pair, with the same
initial conditions.

Assume these pairs are nonexplosive and that the hypotheses
of Theorem 3.1 of \cite{anderson2012efficient} hold for the
split-coupled pair, with constants independent of $\varepsilon$.
Then, for each fixed $T>0$, there exists a constant $C_T>0$
such that
\begin{equation*}
\E \left[ \sup_{t\leq T}
\left(
f\left(X_{\mathrm{MSC}}^{\alpha_\varepsilon}(t)\right)
-
f\left(X_{\mathrm{MSC}}^{\beta_\varepsilon}(t)\right)
\right)^2\right]
\leq C_T\varepsilon.
\end{equation*}
In particular,
\begin{equation*}
\operatorname{Var}\left(
f\left(X_{\mathrm{MSC}}^{\alpha_\varepsilon}(T)\right)
-
f\left(X_{\mathrm{MSC}}^{\beta_\varepsilon}(T)\right)
\right)
=
O(\varepsilon).
\end{equation*}
\end{corollary}

\begin{proof}
By Theorem~\ref{equal theorem111}, after relabeling the two coordinates if necessary, the pair
$\left(X_{\mathrm{MSC}}^{\alpha_\varepsilon},X_{\mathrm{MSC}}^{\beta_\varepsilon}\right)$
has the same law as the corresponding split-coupled pair.  The result then follows directly from Theorem 3.1 of \cite{anderson2012efficient}, applied with perturbation size of order $\varepsilon$.
\end{proof}

We next state the version needed for wider stencils and higher-order derivative formulas.

\begin{corollary}
\label{cor:msc-contrast-variance}
Fix an integer $J \geq 2$ and coefficients $c_1,\dots,c_J$ satisfying
\begin{equation*}
\sum_{r=1}^J c_r=0.
\end{equation*}
For each sufficiently small $\varepsilon>0$, let
$\theta_1^\varepsilon,\dots,\theta_J^\varepsilon\in\mathbb R^K$
satisfy
\begin{equation*}
\max_{1\leq r,s\leq J}
\left\|\theta_r^\varepsilon-\theta_s^\varepsilon\right\|
\leq C_0\varepsilon
\end{equation*}
for some constant $C_0>0$, and suppose the hypotheses of Corollary~\ref{cor:msc-pairwise-variance} hold uniformly for all pairs of parameter vectors in this finite collection.  Let the parameterized processes
$X_{\mathrm{MSC}}^{\theta_1^\varepsilon},\dots,X_{\mathrm{MSC}}^{\theta_J^\varepsilon}$
be generated by MSC, and define the finite-difference numerator
\begin{equation*}
N_\varepsilon(T)
=
\sum_{r=1}^J c_r f\left(X_{\mathrm{MSC}}^{\theta_r^\varepsilon}(T)\right).
\end{equation*}
Then, for each fixed $T>0$,
\begin{equation*}
\E\left[N_\varepsilon(T)^2\right]=O(\varepsilon),
\qquad
\operatorname{Var}\left(N_\varepsilon(T)\right)=O(\varepsilon).
\end{equation*}
\end{corollary}

\begin{proof}
Since $\sum_{r=1}^J c_r=0$, we may rewrite $N_\varepsilon(T)$ using the first process as a reference:
\begin{align*}
N_\varepsilon(T)
&=
\sum_{r=2}^J
c_r
\left(
f\left(X_{\mathrm{MSC}}^{\theta_r^\varepsilon}(T)\right)
-
f\left(X_{\mathrm{MSC}}^{\theta_1^\varepsilon}(T)\right)
\right).
\end{align*}
For each $r$, the parameter vectors $\theta_r^\varepsilon$ and $\theta_1^\varepsilon$ differ by $O(\varepsilon)$.  Corollary~\ref{cor:msc-pairwise-variance} therefore gives
\begin{equation*}
\E\left[\left(
f\left(X_{\mathrm{MSC}}^{\theta_r^\varepsilon}(T)\right)
-
f\left(X_{\mathrm{MSC}}^{\theta_1^\varepsilon}(T)\right)
\right)^2\right]
=
O(\varepsilon).
\end{equation*}
Since there are only finitely many terms in the sum, the bound
\begin{equation*}
\left(\sum_{r=2}^J u_r\right)^2
\leq
(J-1)\sum_{r=2}^J u_r^2
\end{equation*}
implies $\E\left[ N_\varepsilon(T)^2\right] =O(\varepsilon)$.  The variance bound follows immediately from
\begin{equation*}
\operatorname{Var}\left(N_\varepsilon(T)\right)
\leq
\E\left[N_\varepsilon(T)^2\right].
\end{equation*}
\end{proof}

The preceding corollaries provide the variance-order estimates needed for the later bias-variance analysis. In finite-state settings, a $Q$-matrix perturbation argument gives sharper information: the mean and second moment of a coupled finite-difference numerator admit first-order asymptotic expansions.
We record this as a finite-state refinement, which can be useful for understanding the leading constants underlying the order bounds, choosing perturbation sizes, and comparing coupling methods beyond order estimates.

\begin{theorem}
\label{thm:finite-state-generator-expansion}
Let  $\{Z^\varepsilon\}_{\varepsilon \ge 0 }$ be a family of continuous-time Markov chains on a finite state space $E$, 
with common initial distribution $\mu_0$ and respective $Q$-matrices $Q^\varepsilon$.
Suppose there exists a matrix $G$ such that
\begin{equation*}
Q^\varepsilon = Q^0+\varepsilon G+R^\varepsilon,
\qquad \text{where}\qquad  
|R^\varepsilon|=o(\varepsilon),
\end{equation*}
as $\varepsilon \downarrow 0$,
where $|\cdot|$ denotes any matrix norm. Suppose also that there is a subset $\Delta\subset E$ such that $\mu_0$ is supported on $\Delta$ and $\Delta$ is invariant for the chain with $Q$-matrix $Q^0$. Let $H:E\to\mathbb R$ satisfy $H(x)=0$ for all $x\in\Delta$. Then, for each fixed $T>0$,
\begin{align*}
\E \left[ H\left(Z^\varepsilon(T)\right)\right]
&= a_T\varepsilon+o(\varepsilon),\\
\E\left[H\left(Z^\varepsilon(T)\right)^2\right]
&=
b_T\varepsilon+o(\varepsilon),
\end{align*}
where $a_T\in\mathbb R$ and $b_T\geq 0$ are given by
\begin{align*}
a_T &= \mu_0 \left(
\int_0^T e^{(T-s)Q^0}G e^{sQ^0}\,ds
\right)H,\\
b_T &= \mu_0 \left(\int_0^T e^{(T-s)Q^0}G e^{sQ^0}\,ds \right)H^2,
\end{align*}
and $H^2$ denotes the pointwise square of $H$. Consequently,
\begin{equation*}
\operatorname{Var}\left(H\left(Z^\varepsilon(T)\right)\right) = b_T\varepsilon+o(\varepsilon).
\end{equation*}
\end{theorem}

\begin{proof}
Since $E$ is finite, the transition matrix of $Z^\varepsilon$ at time $T$ is
\begin{equation*}
P_T^\varepsilon=e^{TQ^\varepsilon}.
\end{equation*}
The standard first-order expansion of the matrix exponential gives
\begin{equation*}
e^{TQ^\varepsilon}
= e^{TQ^0} + \varepsilon
\int_0^T e^{(T-s)Q^0}G e^{sQ^0}\,ds + o(\varepsilon).
\end{equation*}
Identifying functions on $E$ with column vectors and $\mu_0$ with a row vector, for any function $\varphi:E\to\mathbb R$,
\begin{equation*}
\E\left[ \varphi\left(Z^\varepsilon(T)\right)\right]
= \mu_0 e^{TQ^\varepsilon}\varphi
=
\mu_0 e^{TQ^0}\varphi
+ \varepsilon
\mu_0
\left(
\int_0^T e^{(T-s)Q^0}G e^{sQ^0}\,ds
\right)\varphi + o(\varepsilon).
\end{equation*}
Because $\mu_0$ is supported on $\Delta$, $\Delta$ is invariant for the chain with $Q$-matrix $Q^0$, and both $H$ and $H^2$ vanish on $\Delta$, we have
\begin{equation*}
\mu_0 e^{TQ^0}H=0,
\qquad \text{and}\qquad
\mu_0 e^{TQ^0}H^2=0.
\end{equation*}
Applying the preceding expansion with $\varphi=H$ and then with $\varphi=H^2$ gives the stated formulas for $a_T$ and $b_T$. Finally,
\begin{equation*}
\left(\E\left[H\left(Z^\varepsilon(T)\right)\right]\right)^2=O(\varepsilon^2),
\end{equation*}
and hence
\begin{align*}
\operatorname{Var}\left(H\left(Z^\varepsilon(T)\right)\right)
&= \E\left[H\left(Z^\varepsilon(T)\right)^2\right] - \left(\E\left[H\left(Z^\varepsilon(T)\right)\right]\right)^2= b_T\varepsilon+o(\varepsilon).
\end{align*}
\end{proof}

Theorem~\ref{thm:finite-state-generator-expansion} applies directly to finite-difference estimators generated by MSC. For example, for some $j\in\{1,\dots,K\}$, let
\begin{equation*}
Z^\varepsilon
=
\left(
X^{\theta+a_1\varepsilon e_j},
\dots,
X^{\theta+a_J \varepsilon e_j}
\right),
\qquad
H(x_1,\dots,x_J)
=
\sum_{r=1}^J c_r f(x_r),
\end{equation*}
where $\sum_{r=1}^J c_r=0$, and suppose each coordinate process has finite state space $S$. If the coordinate processes share the same initial condition, then the initial law of $Z^\varepsilon$ is supported on the diagonal
\begin{equation*}
\Delta=\{(x,\dots,x):x\in S\}.
\end{equation*}
Moreover, $H$ vanishes on $\Delta$, and when $\varepsilon=0$, all parameter vectors equal $\theta$, so the MSC construction leaves $\Delta$ invariant. Thus, whenever the $Q$-matrix of $Z^\varepsilon$ has the first-order expansion assumed in Theorem~\ref{thm:finite-state-generator-expansion}, the theorem yields leading-order expansions for the mean and second moment of the finite-difference numerator.

\begin{remark}
\label{remark:phosphorylation}
In the subsequent application sections, the processive phosphorylation network~\eqref{djdichfn} has a finite reachable state space by conservation of the total amounts of substrate, kinase, and phosphatase. On this state space, the mass-action intensity functions are Lipschitz, and because the perturbed parameters enter as rate constants, the intensities and drift change by $O(\varepsilon)$ under the finite-difference perturbations used below. Thus the hypotheses needed for the variance-order estimates are satisfied, and the network also lies in the finite-state setting of Theorem~\ref{thm:finite-state-generator-expansion}.
\end{remark}

For comparison, we also record the corresponding finite-state variance limit for independently generated paths, which will be used in the RMSE analysis of the next section.

\begin{theorem}
\label{thm:indep-finite-state-variance}
Let $\{X^\theta:\theta\in\mathbb R^K\}$ be a family of continuous-time Markov chains on a finite state space $E$, with common initial distribution $\mu_0$ and respective $Q$-matrices $Q^\theta$. Fix $\theta\in\mathbb R^K$, an integer $J \geq 2$, coefficients
$c_1,\dots,c_J\in\mathbb R$, and a function $f:E\to\mathbb R$.  For each $r\in\{1,\dots,J\}$, let $\theta_r^\varepsilon$ satisfy
\begin{equation*}
Q^{\theta_r^\varepsilon}\longrightarrow Q^\theta
\qquad\text{as }\varepsilon\downarrow0.
\end{equation*}
Let
$X_{\mathrm{Indep}}^{\theta_1^\varepsilon},\dots,
X_{\mathrm{Indep}}^{\theta_J^\varepsilon}$
be generated independently, and define
\begin{equation*}
N_\varepsilon(T)
=
\sum_{r=1}^J
c_r f\left(X_{\mathrm{Indep}}^{\theta_r^\varepsilon}(T)\right).
\end{equation*}
Then, for each fixed $T>0$,
\begin{equation*}
\operatorname{Var}\left(N_\varepsilon(T)\right)
=
\left(\sum_{r=1}^J c_r^2\right)
\operatorname{Var}\left(f\left(X^\theta(T)\right)\right)
+
o(1).
\end{equation*}
\end{theorem}

\begin{proof}
By independence,
\begin{equation*}
\operatorname{Var}\left(N_\varepsilon(T)\right)
=
\sum_{r=1}^J
c_r^2
\operatorname{Var}\left(
f\left(X_{\mathrm{Indep}}^{\theta_r^\varepsilon}(T)\right)
\right).
\end{equation*}
Since $E$ is finite and
$Q^{\theta_r^\varepsilon}\to Q^\theta$, continuity of the matrix exponential gives
\begin{equation*}
e^{TQ^{\theta_r^\varepsilon}}
\longrightarrow
e^{TQ^\theta}.
\end{equation*}
Hence the first and second moments of
$f(X_{\mathrm{Indep}}^{\theta_r^\varepsilon}(T))$
converge to those of $f(X^\theta(T))$, and therefore
\begin{equation*}
\operatorname{Var}\left(
f\left(X_{\mathrm{Indep}}^{\theta_r^\varepsilon}(T)\right)
\right)
=
\operatorname{Var}\left(f\left(X^\theta(T)\right)\right)
+
o(1).
\end{equation*}
Summing over $r$ gives the result.
\end{proof}

\section{General optimal RMSE scaling for finite-difference estimators}
\label{General optimal RMSE scaling for finite-difference estimators}

All finite-difference estimators considered below exhibit the same basic bias-variance tradeoff. They differ only in (i) the order of the finite-difference bias and (ii) the variance scaling of the finite-difference numerator under the chosen coupling method. We therefore carry out the general RMSE optimization once here and summarize the resulting optimal scaling rates for the different coupling methods.

Fix $T>0$, let $f:\mathbb Z_{\geq 0}^d\to\mathbb R$ be an observable of interest, and, for $\theta\in\mathbb R^K$, define $g(\theta) \stackrel{\text { def }}{=} \E\left[f(X^\theta(T))\right].$

Let $D(\theta)$ denote a parameter derivative of $g$ of total order $k\geq1$, fix an integer $J \geq 2$, and let
$e_1,\dots,e_K$ denote the canonical unit vectors in $\mathbb R^K$. For fixed
coefficients $a_{rj}$, define
\begin{equation*}
\theta_r^\varepsilon
=
\theta+\varepsilon\sum_{j=1}^K a_{rj}e_j,
\qquad r=1,\dots,J.
\end{equation*}
We consider finite-difference approximations of the form
\begin{align}
\label{firferofdi}
D_\varepsilon(\theta)
=
\frac{\displaystyle\sum_{r=1}^J c_r g(\theta_r^\varepsilon)}
{\varepsilon^k},
\end{align}
where $c_1,\dots,c_J$ are the finite-difference coefficients and the
coefficients $a_{rj}$ specify the stencil points. In particular, we assume $\sum_{r=1}^J c_r=0.$ The corresponding Monte Carlo estimator is
\begin{align}
\label{estimator fpohujt}
\widehat D_\varepsilon(\theta)
&=
\frac{1}{N}\sum_{i=1}^N d_{[i]}(\varepsilon),
&
d_{[i]}(\varepsilon)
&=
\frac{\displaystyle\sum_{r=1}^J
c_r f\left(X_{[i]}^{\theta_r^\varepsilon}(T)\right)}
{\varepsilon^k}.
\end{align}
For each $i$, the processes
$\{X_{[i]}^{\theta_r^\varepsilon}\}_{r=1}^J$
form one collection of coupled sample paths, and the $N$ collections are
independent and identically distributed.

We measure the accuracy of the estimator using its root mean square error (RMSE),
\begin{align*}
\operatorname{RMSE}(\widehat D_\varepsilon(\theta))
&=
(
\E[
(
\widehat D_\varepsilon(\theta)-D(\theta)
)^2
]
)^{1/2} \\
&=
(
\operatorname{Var}(\widehat D_\varepsilon(\theta))
+
(
\E[\widehat D_\varepsilon(\theta)]-D(\theta)
)^2
)^{1/2}.
\end{align*}
This criterion captures both sources of error relevant here: the bias inherent to the finite-difference estimator and the statistical error due to the variance of the Monte Carlo estimator. In particular, an RMSE below $\eta$ implies that both the absolute bias is below $\eta$ and the variance is below $\eta^2$.

Throughout this section, all little-$o$ terms are with
respect to $\varepsilon\downarrow0$.
We assume that the finite-difference approximation has bias order $b>0$ with a nonzero leading coefficient; that is,
\begin{align}
\begin{split}
\label{ioffog9rtfq}
D_\varepsilon(\theta)-D(\theta)
&=
B\varepsilon^b+o(\varepsilon^b),\\
\E[\widehat D_\varepsilon(\theta)]-D(\theta)
&=
B\varepsilon^b+o(\varepsilon^b),
\end{split}
\end{align}
for some constant $B\neq0$, where the second equality follows from
$\E[\widehat D_\varepsilon(\theta)]=D_\varepsilon(\theta)$.

To treat the different coupling methods within a common framework, define the finite-difference numerator $N_{[i]}^\varepsilon
=
\sum_{r=1}^J
c_r f(X_{[i]}^{\theta_r^\varepsilon}(T)).$
Suppose that, under the coupling method being considered,
\begin{equation}
\label{eq:numerator-variance-scaling}
\operatorname{Var}(N_{[i]}^\varepsilon)
=
C_V\varepsilon^q+o(\varepsilon^q)
\end{equation}
for some $C_V>0$ and $q<2k$. Since the $N$ collections of paths are independent and identically distributed,
\begin{equation*}
\operatorname{Var}(\widehat D_\varepsilon(\theta))
=
C_V N^{-1}\varepsilon^{q-2k}
+
o(N^{-1}\varepsilon^{q-2k}).
\end{equation*}
Moreover, the bias expansion~\eqref{ioffog9rtfq} gives
\begin{equation*}
(
\E[\widehat D_\varepsilon(\theta)]-D(\theta)
)^2
=
B^2\varepsilon^{2b}+o\left(\varepsilon^{2b}\right).
\end{equation*}
Therefore,
\begin{align*}
\E[
(
\widehat D_\varepsilon(\theta)-D(\theta)
)^2
]
&=
\operatorname{Var}(\widehat D_\varepsilon(\theta))
+
(
\E[\widehat D_\varepsilon(\theta)]-D(\theta)
)^2\\
&=
C_V N^{-1}\varepsilon^{q-2k}
+
B^2\varepsilon^{2b}
+
o(N^{-1}\varepsilon^{q-2k})
+
o(\varepsilon^{2b}).
\end{align*}
For fixed $N$, the leading-order terms are minimized at
\begin{equation*}
\varepsilon^*
=
(
\tfrac{(2k-q)C_V}{2 b B^2 N}
)^{\frac{1}{2b+2k-q}}.
\end{equation*}
With $\varepsilon=\varepsilon^*(N)$, the corresponding
optimal mean square error and root mean square error
satisfy, as $N\to\infty$,
\begin{align}
\begin{split}
\label{eq:general-optimal-scaling}
\operatorname{MSE}_{\mathrm{opt}}
&=
O\left(
N^{-\frac{2b}{2b+2k-q}}
\right),\\
\operatorname{RMSE}_{\mathrm{opt}}
&=
O\left(
N^{-\frac{b}{2b+2k-q}}
\right).
\end{split}
\end{align}

It remains to identify the numerator-variance exponent $q$ for each coupling method. In the finite-state setting, under the hypotheses of Theorem~\ref{thm:indep-finite-state-variance}, independently generated paths satisfy
\begin{equation*}
\operatorname{Var}(N_{[i]}^\varepsilon)
=
C_{\mathrm{ind}}+o(1),
\end{equation*}
where $C_{\mathrm{ind}} =
\left(\sum_{r=1}^J c_r^2\right)
\operatorname{Var}\left(f(X^\theta(T))\right).$
 When $C_{\mathrm{ind}}>0$, this corresponds to $q=0$, and~\eqref{eq:general-optimal-scaling} yields
\begin{equation*}
\operatorname{RMSE}_{\mathrm{opt}}
=
O\left(
N^{-\frac{b}{2(b+k)}}
\right).
\end{equation*}

For MSC, Corollary~\ref{cor:msc-contrast-variance} gives the general variance bound
$O(\varepsilon)$, while Theorem~\ref{thm:finite-state-generator-expansion} gives
\begin{equation*}
\operatorname{Var}(N_{[i]}^\varepsilon)
=
C_{\mathrm{cpl}}\varepsilon+o(\varepsilon),
\end{equation*}
provided the hypotheses of that theorem hold. When $C_{\mathrm{cpl}}>0$, this corresponds to $q=1$ and
\begin{equation*}
\operatorname{RMSE}_{\mathrm{opt}}
= O\left(
N^{-\frac{b}{2b+2k-1}}
\right).
\end{equation*}

The resulting optimal RMSE scaling rates are summarized in Table~\ref{table dfgiohrefe}.
\begin{table}
    \caption{Optimal RMSE scaling for the estimator~\eqref{estimator fpohujt} under the different coupling methods}
    \label{table dfgiohrefe}
    \centering
    \small
    \setlength{\tabcolsep}{4pt}
    \renewcommand{\arraystretch}{1.2}
    \begin{tabular}{lccc}
    \toprule
     & \textbf{Independent} & \textbf{SC} & \textbf{MSC}\\
    \midrule
    \textbf{Variance exponent $q$}
     & $0$ & $1$ & $1$\\
    \textbf{Optimal RMSE}
     & $O\left(N^{-\frac{b}{2(b+k)}}\right)$
     & $O\left(N^{-\frac{b}{2b+2k-1}}\right)$
     & $O\left(N^{-\frac{b}{2b+2k-1}}\right)$\\
    \bottomrule
    \end{tabular}
\end{table}

We do not include CRP in the theoretical scaling comparison because, to our knowledge, there is no corresponding small-perturbation variance theory for CRP. Previous work has shown that split coupling can be obtained, in a precise limiting sense, from CRP by allowing increasingly frequent recoupling \cite{anderson2014asymptotic}, while numerical studies have also shown that CRP can lose effectiveness over time and, over intermediate time scales, may perform little better than independent simulation \cite{anderson2012efficient}. We therefore treat CRP empirically in the numerical comparisons below rather than assigning it a theoretical RMSE scaling rate.

\section{Numerical applications}
\label{sec:numerics}

We begin by introducing the processive phosphorylation/dephosphorylation model that we will use for all of our numerical experiments and then present three applications of the multi-path stacked coupling method: simultaneous computation of many first derivatives, higher-order finite-difference estimation of a first derivative, and higher-order derivative estimation. Detailed numerical methodology and application-specific implementation
details are provided in Sections~\ref{sec:protocol}--\ref{sec:Application 3 change}
of the Supplementary Material.

\subsection{The processive phosphorylation/dephosphorylation model}
\label{sec:model_example}

Phosphorylation and dephosphorylation are ubiquitous mechanisms in cellular signaling and regulation, and phosphorylation cycles have also been studied as mechanisms for integral regulation and adaptation \cite{fang2017integral,fang2019adaptation}. For our numerical examples, we use the sequential, processive $n$-site phosphorylation/dephosphorylation network studied by Conradi and Shiu \cite{conradi2015global}. Broader classes of processive multisite phosphorylation networks, including variants with irreversible reactions and product inhibition, were subsequently studied by Eithun and Shiu \cite{eithun2017all}. The network used here is
\begin{align}
\label{djdichfn}
\begin{split}
& S_0+K \underset{\alpha_2}{\stackrel{\alpha_1}{\rightleftharpoons}} S_0 K \underset{\alpha_4}{\stackrel{\alpha_3}{\rightleftharpoons}} S_1 K \underset{\alpha_6}{\stackrel{\alpha_5}{\rightleftharpoons}} \ldots \underset{\alpha_{2n}}{\stackrel{\alpha_{2n-1}}{\rightleftharpoons}} S_{n-1} K \xrightarrow{\alpha_{2 n+1}} S_n+K, \\
& S_n+F \underset{\beta_{2 n}}{\stackrel{\beta_{2 n+1}}{\rightleftharpoons}} S_n F \underset{\beta_{2 n-2}}{\stackrel{\beta_{2 n-1}}{\rightleftharpoons}} \ldots \underset{\beta_4}{\stackrel{\beta_5}{\rightleftharpoons}} S_2 F \underset{\beta_2}{\stackrel{\beta_3}{\rightleftharpoons}} S_1 F \xrightarrow{\beta_1} S_0+F.
\end{split}
\end{align}
The parameter vector is
\begin{align*}
    \theta
    &=\left(\theta_1, \cdots, \theta_{4n+2}\right) =\left(\alpha_1, \cdots, \alpha_{2n+1}, \beta_1, \cdots, \beta_{2n+1}\right).
\end{align*}
Let
\begin{align*}
X^\theta(t)
=
\bigl(&S_0^\theta(t),S_n^\theta(t),K^\theta(t),F^\theta(t),
(S_0K)^\theta(t),\dots,(S_{n-1}K)^\theta(t), (S_1F)^\theta(t),\dots,(S_nF)^\theta(t)\bigr)
\end{align*}
denote the state vector of the chemical reaction network at time $t$. 

For all parameter vectors $\theta$, we use the initial condition
\begin{align*}
&S^{\theta}_0(0)=100, \quad S^{\theta}_n(0)= 0,\quad K^{\theta}(0)=10, \quad F^{\theta}(0)=10, \\
& (S_i K)^{\theta}(0)=0, \quad i\in \{0, \dots, n-1\} ,\quad (S_j F)^{\theta}(0)=0, \quad j\in \{1, \dots, n\}.
\end{align*}
This initial condition is fixed throughout the numerical experiments.
As a representative quantity of interest, we consider the expected number of $S_0$ molecules at a fixed terminal time $T=30$:
\begin{align}
\label{eq:g}
    g(\theta) = \E\left[f\left(X^\theta(T)\right)\right] = \E\left[S_0^\theta(T)\right].
\end{align}

In the numerical experiments below, we compare the performance of all three coupling methods (CRP, SC, and MSC) with the baseline independent simulation (Indep)  for estimating parameter sensitivities of $g\left(\theta\right)$ using the following rate constants
\begin{align*}
&\alpha_1=0.02, \quad \alpha_2=0.2, \quad \alpha_{2n+1}=2.0,\\
&\alpha_{2k+1} = 1, \quad \alpha_{2k+2} = 0.5, \quad k \in \{1,\dots,n-1\}; \\
& \beta_1=1.5,  \quad \beta_{2n}=0.15,\quad \beta_{2n+1}=0.015,\\
&\beta_{2k-2} = 0.4, \quad \beta_{2k-1} = 0.8, \quad k \in \{2,\dots,n\}.
\end{align*}

To keep the subsequent application subsections concise, we note here that the processive $n$-site network~\eqref{djdichfn} has three conserved quantities: 
\begin{align*}
    S^{\theta}_{\mathrm{tot}} &= S_0^\theta + S_n^\theta + \sum_{i = 0}^{n-1} \left(S_i K\right)^\theta + \sum_{i = 1}^{n} \left(S_i F\right)^\theta, \\
    K^{\theta}_{\mathrm{tot}} &= K^\theta + \sum_{i = 0}^{n-1} \left(S_i K\right)^\theta, \qquad 
    F^{\theta}_{\mathrm{tot}} = F^\theta + \sum_{i = 1}^{n} \left(S_i F\right)^\theta. 
\end{align*}
These three conserved quantities together involve every species in the network. Consequently, each species count is bounded by the corresponding conserved total, and the reachable state space is finite.
 Together with the common initial condition and the form of the mass-action intensities, this verifies  the hypotheses needed to apply Theorems~\ref{thm:finite-state-generator-expansion} and~\ref{thm:indep-finite-state-variance}; see also Remark~\ref{remark:phosphorylation}. We therefore apply the conclusions of those results directly in the subsections below.

\subsection{Simultaneous computation of many first derivatives}
\label{Motivation 1}

We first consider the simultaneous estimation of all $K$ first-order parametric sensitivities
$\frac{\partial}{\partial\theta_j}g(\theta)$, $j\in\{1,\dots,K\}$, where $g$ is defined in \eqref{eq:g} with $T = 30$.
For each parameter direction, we use the forward finite-difference approximation introduced in Section~\ref{sec:FD_settings}, which has bias $O(\varepsilon_j)$ and requires simulation of the pair
$\left(X^\theta,X^{\theta+\varepsilon_j e_j}\right)$.

For this application, the key distinction among the methods is the combination of variance reduction and the number of simulated component paths required to estimate all $K$ sensitivities.
 For each parameter direction, Theorem~\ref{equal theorem111} implies that the corresponding MSC and SC finite-difference estimators have the same variance, while previous numerical studies have found that SC often provides greater variance reduction than CRP \cite{anderson2012efficient,srivastava2013comparison}. 
The relevant notion of efficiency is the total computational work required to achieve a given RMSE. Wall-clock time, however, depends strongly on implementation details and hardware.
We therefore use the total number of simulated component paths as a simple proxy for computational effort that is less sensitive to implementation details than wall-clock time.
 Since all methods simulate the same reaction network over the same time interval and at nearby parameter values, this metric provides a natural way to quantify the computational savings obtained by sharing paths across parameter directions.

Under this path-count metric, suppose that $N$ Monte Carlo replications are used for each parameter direction. SC simulates a separate nominal--perturbed pair for each of the $K$ directions and therefore requires $2KN$ component paths. By contrast, each MSC replication jointly simulates the nominal path $X^\theta$ and all $K$ perturbed paths $X^{\theta+\varepsilon_j e_j}$, $j\in\{1,\dots,K\}$, and therefore requires only $(K+1)N$ component paths. CRP and Indep likewise require $(K+1)N$ component paths under this accounting. Thus, MSC has the pairwise variance-reduction properties of SC while requiring the same number of simulated component paths as CRP and Indep.

To assess overall accuracy across all parameter directions, we use the sum of the empirically optimized  RMSEs over the $K$ directions. Following the numerical procedure and application-specific details in
Sections~\ref{sec:protocol} and~\ref{sec:Application 1 change}
of the Supplementary Material, we plot this quantity against the total number of simulated component paths on a base-$2$ log-log scale. See Figure \ref{fig:RMSE_111}. Each plotted value is estimated using $500$ independent repetitions of the entire estimation procedure to reduce Monte Carlo variability.

\begin{figure}
    \centering
    \includegraphics[width=0.45\linewidth]{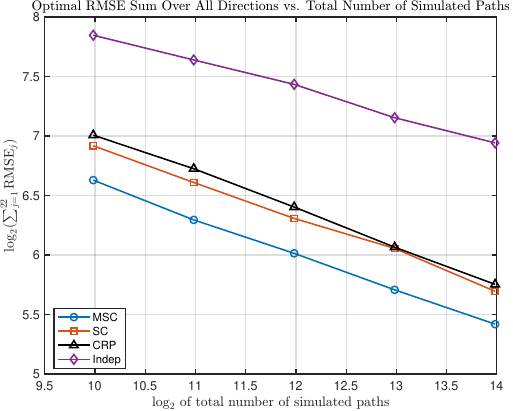}
    \caption{Base-$2$ log-log plot of the sum of the empirically optimized RMSEs over all parameter directions versus the total number of simulated component paths for the $O(\varepsilon_j)$-bias forward finite-difference estimators.}
    \label{fig:RMSE_111}
\end{figure}

\begin{table}
    \centering
    \caption{Theoretical and empirical slopes under different coupling methods.}
    \label{table 111}
    \renewcommand{\arraystretch}{1.2}
    \begin{tabular}{lcccc}
    \toprule
     & \textbf{Indep} & \textbf{CRP} & \textbf{SC} & \textbf{MSC}\\
     \midrule
     \textbf{Theoretical slope} & $-0.2500$ & --- & $-0.3333$ & $-0.3333$  \\ 
     \textbf{Empirical slope} 
& $-0.2263$
& $-0.3137$
& $-0.3060$ 
& $-0.3028$\\ 
    \bottomrule
    \end{tabular}
\end{table}

The empirical slopes for Indep, SC, and MSC closely match the theoretical slopes derived in Section~\ref{General optimal RMSE scaling for finite-difference estimators}. The theoretical and empirical slopes are reported in Table~\ref{table 111}. Specifically, Indep exhibits a slope of approximately $-\frac{1}{4}$, whereas MSC and SC exhibit slopes of approximately $-\frac{1}{3}$. At present, there is no corresponding theoretical slope available for CRP. 
Among the methods considered in this example, MSC achieves the smallest sum of empirically optimized RMSEs over the range of simulated-path budgets considered.

\subsection{Higher-order finite-difference estimation of a first derivative}
\label{Motivation 2}

We next examine whether MSC improves efficiency when higher-order centered finite-difference approximations are used to reduce the bias in computing sensitivities.
For $g$ as defined in \eqref{eq:g} with $T=30$, and for a fixed parameter
$\theta_j$, we consider the centered finite-difference approximations
\eqref{eq:second-order-first-derivative} and
\eqref{eq:fourth-order-first-derivative}, which have bias of order
$O(\varepsilon_j^2)$ and $O(\varepsilon_j^4)$, respectively, and use the
path collections in \eqref{eq:second-order-first-derivative-paths} and
\eqref{eq:wider-stencil-paths}.

For the numerical comparison, we estimate the sensitivity
$\frac{\partial}{\partial\theta_1}g(\theta)$ using both centered finite-difference estimators and compare MSC, CRP, and Indep. Section~\ref{General optimal RMSE scaling for finite-difference estimators} predicts more favorable optimal RMSE scaling for MSC than for Indep for both estimators. We omit SC from this comparison for two reasons. For the two-path stencil, Theorem~\ref{equal theorem111} implies that the SC and MSC constructions have the same joint distribution, and both require two component paths per Monte Carlo replication. Thus, there is no distinction between the two methods for this estimator. For the four-path stencil, there is no unique pairwise split-coupling construction, since the four paths can be grouped into coupled pairs in several different ways, and hence there is no single canonical ``SC'' estimator to use as a benchmark.

We evaluate the methods using the empirically optimized RMSE.
 Following the numerical procedure and application-specific details in
Sections~\ref{sec:protocol} and~\ref{sec:Application 2 change}
of the Supplementary Material, we plot this quantity against the total number of simulated component paths on a base-2 log-log scale. Each plotted value is estimated using $500$ independent repetitions of the complete estimation procedure.

Figure~\ref{fig:RMSE_222} presents the results for both estimators, while the theoretical and empirical slopes are reported in Table~\ref{table 222}.

\begin{figure}
    \centering
    \includegraphics[width=0.45\linewidth]{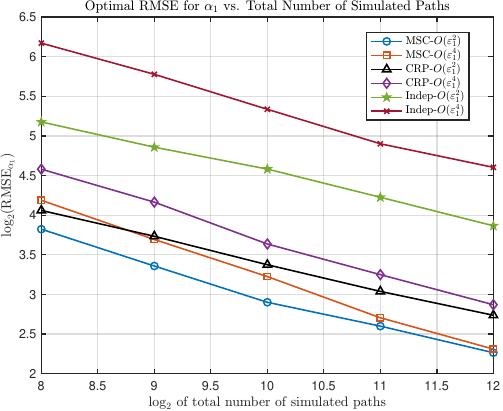}
    \caption{Base-$2$ log-log plot of the RMSEs for the $O(\varepsilon_1^2)$-bias and $O(\varepsilon_1^4)$-bias estimators versus the total number of simulated component paths.}
    \label{fig:RMSE_222}
\end{figure}

\begin{table}
    \centering
    \caption{Theoretical and empirical convergence rates.}
    \label{table 222}
    \renewcommand{\arraystretch}{1.2}
    \begin{tabular}{lccc}
    \toprule
    \textbf{Method}
    &\textbf{Indep}
    &\textbf{CRP}
    &\textbf{MSC}\\
    \midrule
    \multicolumn{4}{c}{\textbf{Centered difference with bias} $O(\varepsilon_j^2)$} \\
    \midrule
    Theoretical rate
    & $-0.3333$
    & --- 
    & $-0.4000$\\
    Empirical rate
    & $-0.3275$
    & $-0.3305$
    & $-0.3896$\\
    \midrule
    
    \multicolumn{4}{c}{\textbf{Centered difference with bias} $O(\varepsilon_j^4)$} \\
    \midrule
    Theoretical rate
    & $-0.4000$
    & --- 
    & $-0.4444$\\
    Empirical rate
    & $-0.3917$
    & $-0.4274$
    & $-0.4692$\\
    \bottomrule
    
    \end{tabular}
\end{table}

The empirical slopes for Indep and MSC are close to the theoretical predictions for both finite-difference estimators. As predicted by the theory, the $O(\varepsilon_1^4)$-bias estimator exhibits a more favorable asymptotic RMSE slope than the $O(\varepsilon_1^2)$-bias estimator. Over the range of path-count budgets considered here, however, this asymptotic advantage has not yet overcome the additional computational cost of the four-point stencil, and the two-point MSC estimator retains a slightly smaller RMSE. For each stencil, MSC achieves the smallest RMSE among the methods considered.

\subsection{Higher-order derivatives}
\label{Motivation 3}

We use the final application to illustrate a central advantage of MSC: it provides a clean framework for estimators that require several coupled paths simultaneously. This is particularly useful for finite-difference estimation of higher-order sensitivities, where the required stencil may involve several nearby parameter values.

For example, the centered finite-difference approximation
\eqref{eq:third-derivative-fd} for a third derivative requires the
four-path collection in \eqref{eq:wider-stencil-paths}.
MSC couples these four processes within a single joint construction. In contrast, pairwise SC does not specify a unique four-path coupling: the four paths can be grouped into coupled pairs in several different ways, leading to different joint constructions. Thus, MSC provides a general-purpose framework that avoids the need to design an application-specific arrangement of pairwise couplings, an advantage that becomes more pronounced as the stencil becomes wider. Specialized methods can provide stronger variance reduction for particular estimators; for example, the double-coupled finite-difference method of Wolf and Anderson \cite{wolf2012finite} is tailored to second-order sensitivities. Our goal here is different: to provide a single coupling framework that applies directly to arbitrary collections of nearby parameter values.

For $g$ as defined in \eqref{eq:g} with $T = 30$, we estimate the third-order sensitivity
$\frac{\partial^3}{\partial\theta_3^3}g(\theta)$
using the centered finite-difference approximation \eqref{eq:third-derivative-fd}, which has bias $O(\varepsilon_3^2)$. We compare the empirically optimized RMSEs obtained using MSC, CRP, and Indep under the same total simulated-path budget. 
Following the numerical procedure and application-specific details in
Sections~\ref{sec:protocol} and~\ref{sec:Application 3 change}
of the Supplementary Material, we plot the RMSE against the total number of simulated component paths on a base-$2$ log-log scale. Each plotted value is estimated using $500$ independent repetitions of the complete estimation procedure.

\begin{figure}
    \centering
    \includegraphics[width=0.45\linewidth]{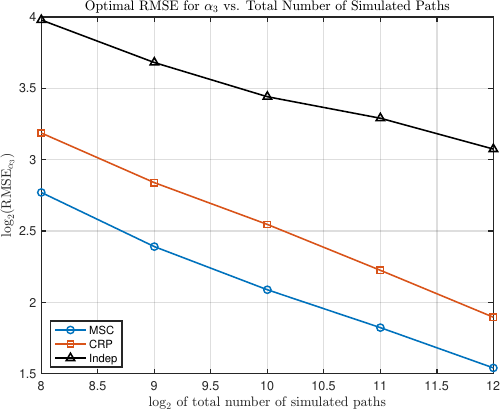}
    \caption{Base-$2$ log-log plot of the empirically optimized RMSE for the $O(\varepsilon_3^2)$-bias third-order sensitivity estimator versus the total number of simulated component paths.}
    \label{fig:RMSE_333}
\end{figure}

\begin{table}
    \centering
    \caption{Theoretical and empirical convergence rates for the empirically optimized RMSE}
    \label{table 333}
    \renewcommand{\arraystretch}{1.2}
    \begin{tabular}{lccc}
    \toprule
    \textbf{Method}
    &\textbf{Indep}
    &\textbf{CRP}
    &\textbf{MSC}\\
    \midrule
    Theoretical rate & $-0.2000$ & --- & $-0.2222$ \\ 
    Empirical rate & $-0.2261$ & $-0.3227$ & $-0.3074$ \\
    \bottomrule
    \end{tabular}
\end{table}

Figure~\ref{fig:RMSE_333} presents the numerical results, and the theoretical and empirical slopes are reported in Table~\ref{table 333}. Over the range of path-count budgets considered, the empirical slopes for Indep and MSC are steeper than the corresponding asymptotic theoretical predictions; no theoretical slope is available for CRP. We believe this discrepancy reflects the fact that the simulations have not yet entered the small-$\varepsilon$ regime underlying the asymptotic theory. Nevertheless, MSC achieves the smallest empirically optimized RMSE among the methods considered throughout the displayed range.

\section{Discussion}
\label{sec:discussion}

We have introduced MSC for jointly simulating a finite
collection of stochastic reaction networks with common
reaction channels and reaction vectors, but possibly different
intensity functions and initial conditions. By
Theorem~\ref{equal theorem111}, every pairwise marginal has
the same law as the corresponding split-coupled pair. Thus,
the full collection can be generated jointly while retaining
the split-coupling behavior of every pair.

We studied sensitivity estimation as one important application,
considering the simultaneous estimation of many first
derivatives, wider finite-difference stencils for a single
first derivative, and higher-order derivatives. In simultaneous
first-derivative estimation, MSC combines the path-count
efficiency of the common reaction path method with the
pairwise variance behavior of split coupling. Under the stated
assumptions, Corollary~\ref{cor:msc-contrast-variance} provides
variance bounds for multi-path finite-difference contrasts,
while the finite-state theory supplies the asymptotic
expansions used in the RMSE analysis. Numerical experiments
on a processive phosphorylation network support this
analysis: MSC achieved the smallest empirically optimized
RMSE among the methods considered in each of the three
applications over the tested path-count budgets.

\vspace{.1in}

\noindent \textbf{Acknowledgments}

DFA gratefully acknowledges financial support from the Office of the Vice Chancellor for Research at UW--Madison, with funding from the Wisconsin Alumni Research Foundation; from the Trustees of the William F. Vilas Estate; and from NSF grant DMS-2051498.

ChatGPT was used for sentence-level editing and to provide initial coding templates for the numerical simulations.

\bibliographystyle{siamplain}
\bibliography{references}

\clearpage

\setcounter{page}{1}
\renewcommand{\thepage}{S\arabic{page}}

\setcounter{section}{0}
\renewcommand{\thesection}{S\arabic{section}}
\renewcommand{\theHsection}{S\arabic{section}}

\setcounter{equation}{0}
\renewcommand{\theequation}{S\arabic{equation}}
\renewcommand{\theHequation}{S\arabic{equation}}

\setcounter{table}{0}
\renewcommand{\thetable}{S\arabic{table}}
\renewcommand{\theHtable}{S\arabic{table}}

\begin{center}
{\Large\bfseries Supplementary Material}

\vspace{0.15in}

{\large
A general-purpose multi-path stacked coupling for stochastic reaction networks}

\vspace{0.12in}

David F.\ Anderson\footnote{anderson@math.wisc.edu} and Jingyi Ma\footnote{jma276@wisc.edu}\\
University of Wisconsin-Madison
\end{center}

\vspace{0.2in}

\section{Numerical methodology}
\label{sec:protocol}

We fix $n=5$ in the model \eqref{djdichfn} throughout the numerical experiments. For each method, let $N$ denote the number of independent Monte Carlo replications averaged to form the sensitivity estimator. If one replication requires $p_m$ component paths under method $m$, then the total number of simulated paths is $Np_m$.

For example, in Application~1, MSC, CRP, and Indep require $K+1$ component paths per replication, while SC requires $2K$ component paths per replication. Thus, $N$ replications correspond to $(K+1)N$ simulated paths for MSC, CRP, and Indep, and $2KN$ simulated paths for SC.  In Application~2, the $O(\varepsilon_1^2)$-bias estimator requires $2$ component paths per replication, while the $O(\varepsilon_1^4)$-bias estimator requires $4$ component paths per replication. Thus, $N$ replications correspond to $2N$ and $4N$ simulated paths, respectively. In Application~3, the centered finite-difference estimator for the third derivative requires $4$ component paths per replication, so $N$ replications correspond to $4N$ simulated paths.

\begin{enumerate}
    \item \textbf{Numerical reference values.}
    \label{Step $1$: Obtaining high-accuracy numerical reference values}

    For each target sensitivity $D_j(\theta)$, we compute a numerical reference value $D_{j,\mathrm{ref}}$ using MSC with a large Monte Carlo replication count, $N_{\mathrm{ref}}=2^{17}$, and a small application-specific perturbation size. The finite-difference estimator and perturbation size used to construct the reference value are specified in the corresponding application-specific section below.

    \item \textbf{Empirical calibration of the perturbation size.}
    \label{Step $2$: Determine the optimal perturbation sizes}

    For each coupling method $m$, target sensitivity $D_j(\theta)$, and Monte Carlo replication count  $N_0=2^8$, we consider perturbations of the form
    \begin{equation}
    \label{vjofjooefoed}
    \varepsilon_{m,j}(C)
    =
    C|\theta_j|N_0^{-\gamma_m},
    \end{equation}
    where $\gamma_m$ denotes the exponent appropriate to the coupling method and finite-difference estimator under consideration, as determined by Section~\ref{General optimal RMSE scaling for finite-difference estimators} when such a theoretical result is available. For CRP, for which we have no corresponding theory, the exponent used in the numerical experiments is treated only as an empirical tuning convention.

    For each candidate value of $C$, we independently repeat the complete estimation procedure based on $N_0$ replications $L_{\mathrm{cal}}=20$ times. If
    $\widehat D_{m,j}^{(\ell)}(C;N_0)$
    denotes the estimate from repetition $\ell$, we compute
    \begin{equation*}
    \widehat{\operatorname{RMSE}}_{m,j}(C;N_0)
    =
    \left[
    \frac{1}{L_{\mathrm{cal}}}
    \sum_{\ell=1}^{L_{\mathrm{cal}}}
    \left(
    \widehat D_{m,j}^{(\ell)}(C;N_0)
    -
    D_{j,\mathrm{ref}}
    \right)^2
    \right]^{1/2}.
    \end{equation*}
    We denote by $\widehat C_{m,j}$ the candidate value that minimizes this empirical RMSE. For other Monte Carlo replication counts $N$, the perturbation is then chosen as
    \begin{equation}
    \label{dpogy}
    \varepsilon_{m,j}(N)
    =
    \widehat C_{m,j}|\theta_j|N^{-\gamma_m}.
    \end{equation}

    \item \textbf{Evaluation at fixed path-count budgets.}
    \label{Step $3$: Estimate the parameter sensitivity using each coupling method}

    Let $p_m$ denote the number of component paths required to generate one Monte Carlo replication for the method and finite-difference estimator under consideration.
    For a fixed total path-count budget $B$, the Monte Carlo replication count is
    \begin{equation*}
    N_m(B)=\frac{B}{p_m},
    \end{equation*}
    where the budgets are selected so that $N_m(B)$ is an integer for every method being compared. Each method is then evaluated using $N_m(B)$ replications and the perturbation size in \eqref{dpogy}.

    \item \textbf{Estimation of the RMSE.}
    \label{Step $4$: Plot empirical optimal RMSE versus computational cost}

    For each method and path-count budget, we repeat the complete estimation procedure
    $L_{\mathrm{eval}}=500$ times. The reported RMSE for sensitivity $D_j(\theta)$ is
    \begin{equation*}
    \widehat{\operatorname{RMSE}}_{m,j}(B)
    =
    \left[
    \frac{1}{L_{\mathrm{eval}}}
    \sum_{\ell=1}^{L_{\mathrm{eval}}}
    \left(
    \widehat D_{m,j}^{(\ell)}(B)
    -
    D_{j,\mathrm{ref}}
    \right)^2
    \right]^{1/2}.
    \end{equation*}
    The resulting RMSEs are plotted against the relevant path-count budget on a base-$2$ log-log scale.
\end{enumerate}

\section{Application 1: implementation details}
\label{sec:Application 1 change}

In Application~1, with $g$ defined in \eqref{eq:g}, the target quantities are all $K=22$ first-order sensitivities
\begin{equation*}
D_j(\theta)
=
\frac{\partial}{\partial\theta_j}g(\theta),
\qquad
j\in\{1,\dots,K\}.
\end{equation*}
The estimators being compared use the forward finite-difference approximation introduced in Section~\ref{sec:FD_settings}.

For the numerical reference values, we use the centered finite-difference estimator
\begin{equation*}
\widehat D_{j,\mathrm{ref}}
=
\frac{1}{N_{\mathrm{ref}}}
\sum_{i=1}^{N_{\mathrm{ref}}}
\frac{
f\left(X_{[i]}^{\theta+h_j^{\mathrm{ref}}e_j}(T)\right)
-
f\left(X_{[i]}^{\theta-h_j^{\mathrm{ref}}e_j}(T)\right)
}{
2h_j^{\mathrm{ref}}
},
\end{equation*}
where
\begin{equation*}
N_{\mathrm{ref}}=2^{17},
\qquad
h_j^{\mathrm{ref}}=0.005|\theta_j|.
\end{equation*}
The required paths for all parameter directions are generated jointly using MSC. The resulting reference values are checked for stability by repeating the calculation with
$h_j^{\mathrm{ref}}/2$. The sum of the absolute changes in the $22$ reference sensitivities was $3.53$, approximately $8.3\%$ of the smallest aggregate RMSE reported in Application~1. We therefore regard the reference values as sufficiently stable for the numerical comparison.

The perturbation exponents and the number of component paths required for one Monte Carlo replication are summarized below.

\begin{table}[H]
\centering
\caption{Perturbation exponents and component paths per replication for Application~1.}
\renewcommand{\arraystretch}{1.2}
\begin{tabular}{lcc}
\toprule
Method
&
$\gamma_m$
&
$p_m$
\\
\midrule
MSC
&
$\frac{1}{3}$
&
$K+1=23$
\\
SC
&
$\frac{1}{3}$
&
$2K=44$
\\
CRP
&
$\frac{1}{4}$
&
$K+1=23$
\\
Indep
&
$\frac{1}{4}$
&
$K+1=23$
\\
\bottomrule
\end{tabular}
\end{table}

For CRP, for which Section~\ref{General optimal RMSE scaling for finite-difference estimators} provides no corresponding theoretical scaling result, we use $\gamma_{\mathrm{CRP}}=1/4$ as an empirical tuning convention.

For each method $m$ and parameter direction $j$, we calibrate
$\widehat C_{m,j}$
at $N_0=2^8$ using the candidate values
\begin{equation*}
C\in\{0.1,0.2,\dots,15.0\}
\end{equation*}
and the procedure in Section~\ref{sec:protocol}. Thus, for the final simulations,
\begin{equation*}
\varepsilon_{m,j}(N)
=
\widehat C_{m,j}|\theta_j|N^{-\gamma_m}.
\end{equation*}

To compare the methods at identical total path-count budgets, we use
\begin{equation*}
B_k
=
(K+1)(2K)2^k,
\qquad
k\in\{0,\dots,4\}.
\end{equation*}
This choice is a common multiple of the $K+1$ paths required by MSC, CRP, and Indep and the $2K$ paths required by SC. Consequently,
\begin{align*}
N_{\mathrm{MSC}}(B_k)
=
N_{\mathrm{CRP}}(B_k)
&=
N_{\mathrm{Indep}}(B_k)
=
(2K)2^k,\\
N_{\mathrm{SC}}(B_k)
&=
(K+1)2^k.
\end{align*}

For each method and budget, we estimate the componentwise RMSEs using
$L_{\mathrm{eval}}=500$ independent repetitions and report the aggregate quantity
\begin{equation*}
\mathcal R_m(B_k)
=
\sum_{j=1}^K
\widehat{\operatorname{RMSE}}_{m,j}(B_k),
\end{equation*}
which is the quantity plotted in Figure~\ref{fig:RMSE_111}.

\section{Application 2: implementation details}
\label{sec:Application 2 change}

In Application~2, with $g$ defined in \eqref{eq:g}, the target quantity is the first-order sensitivity
\begin{equation*}
D(\theta)
=
\frac{\partial}{\partial\theta_1}g(\theta).
\end{equation*}
The estimators being compared are the centered finite-difference approximations with bias of order $O(\varepsilon_1^2)$ and $O(\varepsilon_1^4)$. 

For the numerical reference value, we use the centered finite-difference estimator
\begin{equation*}
\widehat D_{\mathrm{ref}}
=
\frac{1}{12N_{\mathrm{ref}}h^{\mathrm{ref}}}
\sum_{i=1}^{N_{\mathrm{ref}}}
\begin{aligned}[t]
\Bigl[
&-f\left(X_{[i]}^{\theta+2h^{\mathrm{ref}}e_1}(T)\right)
+8f\left(X_{[i]}^{\theta+h^{\mathrm{ref}}e_1}(T)\right)\\
&-8f\left(X_{[i]}^{\theta-h^{\mathrm{ref}}e_1}(T)\right)
+f\left(X_{[i]}^{\theta-2h^{\mathrm{ref}}e_1}(T)\right)
\Bigr],
\end{aligned}
\end{equation*}
where
\begin{equation*}
N_{\mathrm{ref}}=2^{17},
\qquad
h^{\mathrm{ref}}=0.1|\theta_1|.
\end{equation*}
The four required paths are generated jointly using MSC. The resulting reference value is checked for stability by repeating the calculation with
$h^{\mathrm{ref}}/2$. The resulting change in the reference value was $0.046$, approximately $0.95\%$ of the smallest RMSE reported in Application~2.   We therefore regard the reference value as sufficiently stable for the numerical comparison.

The perturbation exponents $\gamma_{m,b}$ and the numbers of component paths per Monte Carlo replication $p_{m,b}$ are summarized below.

\begin{table}[H]
\centering
\caption{Perturbation exponents and component paths per replication for Application~2.}
\renewcommand{\arraystretch}{1.2}
\begin{tabular}{lccc}
\toprule
Method
&
Bias order
&
$\gamma_{m,b}$
&
$p_{m,b}$
\\
\midrule
MSC
&
$O(\varepsilon_1^2)$
&
$\frac{1}{5}$
&
$2$
\\
MSC
&
$O(\varepsilon_1^4)$
&
$\frac{1}{9}$
&
$4$
\\
CRP
&
$O(\varepsilon_1^2)$
&
$\frac{1}{6}$
&
$2$
\\
CRP
&
$O(\varepsilon_1^4)$
&
$\frac{1}{10}$
&
$4$
\\
Indep
&
$O(\varepsilon_1^2)$
&
$\frac{1}{6}$
&
$2$
\\
Indep
&
$O(\varepsilon_1^4)$
&
$\frac{1}{10}$
&
$4$
\\
\bottomrule
\end{tabular}
\end{table}

For CRP, for which Section~\ref{General optimal RMSE scaling for finite-difference estimators} provides no corresponding theoretical scaling result, we use $\gamma_{\mathrm{CRP,2}}=1/6$ for $O(\varepsilon_1^2)$ and $\gamma_{\mathrm{CRP,4}}=1/10$ for $O(\varepsilon_1^4)$ as empirical tuning conventions.

For each method $m$ and bias order $b\in\{2,4\}$, we calibrate
the corresponding constant $\widehat C_{m,b}$  at $N_0=2^8$ using the candidate values
\begin{equation*}
C\in\{0.05,0.1,\dots,3.0\}
\end{equation*}
and the procedure in Section~\ref{sec:protocol}. Thus, for the final simulations,
\begin{equation*}
\varepsilon_{m,b}(N)
=
\widehat C_{m,b}|\theta_1|N^{-\gamma_{m,b}}.
\end{equation*}
To compare the methods at identical total path-count budgets, we use
\begin{equation*}
B_k
=
2^{k},  \qquad k \in \{8,\cdots,12\}.
\end{equation*}
Consequently, for each method $m\in\{\mathrm{MSC},\mathrm{CRP},\mathrm{Indep}\}$,
\begin{align*}
N_{m,2}(B_k) &= \frac{B_k}{2}=2^{k-1},\\
N_{m,4}(B_k) &= \frac{B_k}{4}=2^{k-2},
\end{align*}
where the second subscript denotes the bias order of the finite-difference estimator.
For each method, estimator and budget, we estimate the RMSE using
$L_{\mathrm{eval}}=500$ independent repetitions and report the quantity
\begin{equation*}
\widehat{\operatorname{RMSE}}_{m,b}\left( B_k\right),
\end{equation*}
which is the quantity plotted in Figure~\ref{fig:RMSE_222}.

\section{Application 3: implementation details}
\label{sec:Application 3 change}

In Application~3, with $g$ defined in \eqref{eq:g}, the target quantity is the third-order sensitivity
\begin{equation*}
D(\theta)
=
\frac{\partial^3}{\partial \theta_3^3}g(\theta).
\end{equation*}
We use the centered finite-difference approximation \eqref{eq:third-derivative-fd}, which has bias of order $O(\varepsilon_3^2)$.

For the numerical reference value, we use the centered finite-difference estimator
\begin{equation*}
\widehat D_{\mathrm{ref}}
=
\frac{1}{2N_{\mathrm{ref}}\left(h^{\mathrm{ref}}\right)^3}
\sum_{i=1}^{N_{\mathrm{ref}}}
\begin{aligned}[t]
\Bigl[
&f\left(X_{[i]}^{\theta+2h^{\mathrm{ref}}e_3}(T)\right)
-2f\left(X_{[i]}^{\theta+h^{\mathrm{ref}}e_3}(T)\right)\\
&+2f\left(X_{[i]}^{\theta-h^{\mathrm{ref}}e_3}(T)\right)
-f\left(X_{[i]}^{\theta-2h^{\mathrm{ref}}e_3}(T)\right)
\Bigr],
\end{aligned}
\end{equation*}
where
\begin{equation*}
N_{\mathrm{ref}}=2^{17},
\qquad
h^{\mathrm{ref}}=0.2|\theta_3|.
\end{equation*}
The four required paths are generated jointly using MSC. The resulting reference value is checked for stability by repeating the calculation with
$h^{\mathrm{ref}}/2$. The resulting change in the reference value was $0.069$, approximately $2.36\%$ of the smallest RMSE reported in Application~3. We therefore regard the reference value as sufficiently stable for the numerical comparison.

The perturbation exponents and the number of component paths required for one Monte Carlo replication are summarized below.

\begin{table}[H]
\centering
\caption{Perturbation exponents and component paths per replication for Application~3.}
\renewcommand{\arraystretch}{1.2}
\begin{tabular}{lcc}
\toprule
Method
&
$\gamma_m$
&
$p_m$
\\
\midrule
MSC
&
$\frac{1}{9}$
&
$4$
\\
CRP
&
$\frac{1}{10}$
&
$4$
\\
Indep
&
$\frac{1}{10}$
&
$4$
\\
\bottomrule
\end{tabular}
\end{table}

For CRP, for which Section~\ref{General optimal RMSE scaling for finite-difference estimators} provides no corresponding theoretical scaling result, we use $\gamma_{\mathrm{CRP}}=1/10$ as an empirical tuning convention.

For each method $m$, we calibrate
$\widehat C_{m}$
at $N_0=2^8$ using the candidate values
\begin{equation*}
C\in\{0.01,0.02,\dots,2.00\}
\end{equation*}
and the procedure in Section~\ref{sec:protocol}. Thus, for the final simulations,
\begin{equation*}
\varepsilon_{m}(N)
=
\widehat C_{m}|\theta_3|N^{-\gamma_{m}}.
\end{equation*}
To compare the methods at identical total path-count budgets, we use
\begin{equation*}
B_k
=
2^{k},  \qquad k \in \{8,\cdots,12\}.
\end{equation*}
Consequently,
\begin{align*}
N_{\mathrm{MSC}}(B_k)
= N_{\mathrm{CRP}}(B_k)
= N_{\mathrm{Indep}}(B_k)
= \frac{2^{k}}{4}
= 2^{k-2},
\qquad k\in\{8,\cdots,12\}.
\end{align*}
For each method and budget, we estimate the RMSE using
$L_{\mathrm{eval}}=500$ independent repetitions and report the quantity
\begin{equation*}
\widehat{\operatorname{RMSE}}_{m}\left( B_k\right),
\end{equation*}
which is the quantity plotted in Figure~\ref{fig:RMSE_333}.

\end{document}